\documentclass[runningheads,envcountsect,envcountsame]{llncs}
\usepackage{graphicx}
\usepackage{framed, amsmath, amsfonts, amssymb, enumitem}
\usepackage{balance,complexity,hyperref,pgfplots,tikz,tikz-3dplot,float}
\usetikzlibrary{shapes,positioning}

\usepackage[sort,nocompress]{cite}
\usepackage{subcaption}

\usepackage[noend]{algpseudocode}
\usepackage{algorithm}

\usepackage{comment}

\newtheorem{proc}[theorem]{Procedure}

\usepackage{microtype}

\usetikzlibrary{snakes}

\def\ca#1{{\mathcal{#1}}}
\def\apxmark{$\!\!${\bf(*)}}

\usepackage{environ}

\newcommand{\repeattheorem}[1]{%
	\begingroup
	\renewcommand{\thetheorem}{\ref{#1}}%
	\expandafter\expandafter\expandafter\theorem
	\csname reptheorem@#1\endcsname
	\endtheorem
	\endgroup
}

\NewEnviron{reptheorem}[1]{%
	\global\expandafter\xdef\csname reptheorem@#1\endcsname{%
		\unexpanded\expandafter{\BODY}%
	}%
	\expandafter\theorem\BODY\unskip\label{#1}\endtheorem
}

\newcommand{\repeatlemma}[1]{%
	\begingroup
	\renewcommand{\thelemma}{\ref{#1}}%
	\expandafter\expandafter\expandafter\lemma
	\csname replemma@#1\endcsname
	\endlemma
	\endgroup
}

\NewEnviron{replemma}[1]{%
	\global\expandafter\xdef\csname replemma@#1\endcsname{%
		\unexpanded\expandafter{\BODY}%
	}%
	\expandafter\lemma\BODY\unskip\label{#1}\endlemma
}

\newcommand{\repeatcorollary}[1]{%
	\begingroup
	\renewcommand{\thecorollary}{\ref{#1}}%
	\expandafter\expandafter\expandafter\corollary
	\csname repcorollary@#1\endcsname
	\endcorollary
	\endgroup
}

\NewEnviron{repcorollary}[1]{%
	\global\expandafter\xdef\csname repcorollary@#1\endcsname{%
		\unexpanded\expandafter{\BODY}%
	}%
	\expandafter\corollary\BODY\unskip\label{#1}\endcorollary
}

\newcommand{\repeatproposition}[1]{%
	\begingroup
	\renewcommand{\theproposition}{\ref{#1}}%
	\expandafter\expandafter\expandafter\proposition
	\csname repproposition@#1\endcsname
	\endproposition
	\endgroup
}

\NewEnviron{repproposition}[1]{%
	\global\expandafter\xdef\csname repproposition@#1\endcsname{%
		\unexpanded\expandafter{\BODY}%
	}%
	\expandafter\proposition\BODY\unskip\label{#1}\endproposition
}

\begin{document}
	\title{Isomorphism Testing for $T$-graphs in FPT\thanks{This work was supported by the Czech Science Foundation, project no.~20-04567S.}}
	
	\author{Deniz A\u{g}ao\u{g}lu \c{C}a\u{g}{\i}r{\i}c{\i}\inst{1}\orcidID{0000-0002-1691-0434}%
		\and\\
		Petr Hlin\v en\'y\inst{2}\orcidID{0000-0003-2125-1514}}
	\authorrunning{D.~A\u{g}ao\u{g}lu~\c{C}a\u{g}{\i}r{\i}c{\i} and P.~Hlin\v en\'y}
	
	\institute{
		Masaryk University, Brno, Czech Republic
		\email{agaoglu@mail.muni.cz}\\
		\and
		Masaryk University, Brno, Czech Republic
		\email{hlineny@fi.muni.cz}}
	\maketitle              
	\begin{abstract}
		A $T$-graph (a special case of a chordal graph) is the intersection graph of 
		connected subtrees of a suitable subdivision of a fixed tree~$T$. 
		We deal with the isomorphism problem for $T$-graphs which is 
		\emph{GI-complete} in general -- when $T$ is a part of the input and even a star.
		We prove that the $T$-graph isomorphism problem is in FPT when $T$ is the
		fixed parameter of the problem.
		This can equivalently be stated that isomorphism is in FPT for chordal
		graphs of (so-called) bounded leafage.
		While the recognition problem for $T$-graphs is not known to be in FPT
		wrt.~$T$, we do {\em not} need a $T$-representation to be given (a promise is enough).
		To obtain the result, we combine a suitable isomorphism-invariant
		decomposition of $T$-graphs with the classical tower-of-groups algorithm of
		Babai, and reuse some of the ideas of our isomorphism algorithm for
		$S_d$-graphs [MFCS~2020].
		
	\end{abstract}

	\section{Introduction}
	\label{sec:intro}
	
	Two graphs $G$ and $H$ are called \emph{isomorphic}, denoted by $G \simeq
	H$, if there is a bijection $f:V(G) \rightarrow V(H)$ such that for every
	pair $u,v \in V(G)$, $\{u,v\} \in E(G)$ if and only if $\{f(u),f(v)\} \in
	E(H)$.  The well-known \emph{graph isomorphism problem} asks whether two
	input graphs are isomorphic, and it can be solved efficiently for various
	special graph classes
	\cite{AHU,planarLinear,recogIntervalLinear,hellyCARClinearISO,genus,DBLP:journals/networks/Colbourn81}. 
	On the other hand, it is still unknown whether this problem is
	polynomial-time solvable or not (though, it is not expected to be NP-hard) in the general case, 
	and a problem is said to be \emph{GI-complete} if it is polynomial-time equivalent to the
	graph isomorphism.
	
	We now briefly introduce two complexity classes of parameterized problems. 
	Let $k$ be the parameter, $n$ be the input size, $f$ and $g$ be two computable
	functions, and $c$ be some constant.  A decision problem is in the \emph{class FPT}
	(or {\em FPT-time}) if there exists an algorithm solving that problem correctly in
	time $\mathcal{O}(f(k)\cdot n^{c})$.  Similarly, a decision problem is
	in the \emph{class XP} if there exists an algorithm solving that problem correctly
	in time $\mathcal{O}(f(k)\cdot n^{g(k)})$.  Some parameters which
	yield to \emph{FPT}- or \emph{XP}-time algorithms for the graph isomorphism
	problem can be listed as tree-depth \cite{treedepth}, tree-width
	\cite{treewidth}, maximum degree \cite{degree} and genus \cite{genus}.  
	In this paper, we consider the parameterized complexity of the graph
	isomorphism problem for special instances of intersection graphs which we
	introduce next.
	
	The \emph{intersection graph} for a finite family of sets is an undirected
	graph~$G$ where each set is associated with a vertex of $G$, and each pair
	of vertices in $G$ are joined by an edge if and only if the corresponding
	sets have a non-empty intersection.  Chordal and interval graphs are two of
	the most well-known intersection graph classes related to our research.
	
	A graph is \emph{chordal} if every cycle of length more than three
	has a chord.  They are also defined as the intersection graphs of subtrees
	of some (non-fixed) tree~$T$~\cite{chordalityInters}.  Chordal graphs can be recognized
	in linear time, and they have linearly many maximal cliques which can be
	listed in polynomial time \cite{recogChordaLinear}.  Deciding the
	isomorphism of chordal graphs is a {\em GI-complete} problem~\cite{isoChordalGIComp}.
	A graph $G$ is an \emph{interval graph} if it is the intersection graph of
	a set of intervals on the real line. Interval graphs form a subclass of chordal
	graphs.  They can also be recognized in linear time, and interval graph
	isomorphism can be solved in linear time~\cite{recogIntervalLinear}.
	
	A subdivision of a graph $G$ is the operation of replacing selected edge(s)
	of~$G$ by new induced paths (informally, putting new vertices to the middle of an edge).
	For a fixed graph $H$, an {\em$H$-graph} is the intersection graph of
	connected subgraphs of a suitable subdivision of the graph~$H$~\cite{biro},
	and they generalize many types of intersection graphs.  For instance, interval graphs
	are $K_2$-graphs, their generalization called circular-arc graphs are $K_3$-graphs,
	and chordal graphs are the union of $T$-graphs where $T$ ranges over all trees.
	We, however, consider {\em$T$-graphs} where $T$ is a fixed tree.
	Even though chordal graphs can be recognized in
	linear time \cite{recogChordaLinear}, deciding whether a given chordal graph
	is a $T$-graph is NP-complete when $T$ is on the input~\cite{KLAVIK201585}.  
	In \cite{zemanWG}, Chaplick et al.\ gave an \emph{XP}-time algorithm to 
	recognize $T$-graphs parameterized by the size of $T$.
	
	$S_d$-graphs form a subclass of $T$-graphs where $S_d$ is the star with $d$
	rays.  The isomorphism problem for $S_d$-graphs, and therefore for
	$T$-graphs, was shown to be GI-complete \cite{isoChordalGIComp} with $d$ on the input.  
	In \cite{aaolu2019isomorphism}, we have proved by algebraic means that 
	$S_d$-graph isomorphism can be solved in \emph{FPT}-time parameterized by $d$, and then in
	\cite{SdT-graphs2021efficient} we have extended this approach to an
	\emph{XP}-time algorithm for the isomorphism problem of $T$-graphs 
	parameterized by the size of~$T$.  
	We have also considered in \cite{SdT-graphs2021efficient} the special case
	of isomorphism of proper $T$-graphs
	with a purely combinatorial \emph{FPT}-time algorithm.
	
	\paragraph{\textbf{New contribution.}}
	
	In this paper, we show that the graph isomorphism problem for $T$-graphs can be
	solved in \emph{FPT}-time parameterized by the size of $T$.
	Our algorithm does not assume or rely on $T$-representations of the input graphs to be
	given, and in fact it uses only some special properties of $T$-graphs.
	
	Moreover, our result can be equivalently reformulated as an \emph{FPT}-time
	algorithm for testing isomorphism of chordal graphs of bounded leafage,
	where the {\em leafage} of a chordal graph $G$ can be defined as the least
	number of leaves of a tree $T$ such that $G$ is a $T$-graph.
	Since there is only a bounded number of trees $T$ of a given number of
	leaves, modulo subdivisions, the correspondence of the two formulations is obvious.
	
	Highly informally explaining our approach (which is different
	from \cite{SdT-graphs2021efficient}), we~use chordality and
	properties of assumed $T$-representations of input graphs $G$ and~$G'$ 
	to efficiently compute their special hierarchical canonical
	decompositions into so-called fragments (Section~\ref{sec:decom}).
	Each fragment will be an interval graph, and the isomorphism problem of interval
	graphs is well understood.
	Then we use some classical group-computing tools
	(Section~\ref{sec:algebraic}, Babai's tower-of-groups approach) to
	compute possible ``isomorphisms'' between the decompositions of $G$ and of $G'$
	(Section~\ref{sec:mainalg}); each such isomorphism
	mapping between the fragments of the two decompositions, and
	simultaneously between the neighborhood sets of fragments in
	other fragments ``higher up'' in the decomposition.
	
	We remark that the same problem has been independently and
	concurrently solved by Arvind, Nedela, Ponomarenko and Zeman
	\cite{DBLP:journals/corr/abs-2107-10689},%
	\footnote{To be completely accurate, our paper was first time submitted to a
		conference at the beginning of July 2021, and
		\cite{DBLP:journals/corr/abs-2107-10689} appeared on arXiv just two weeks
		later, without mutual influence regarding the algorithms.}
	using different means (by reducing the problem to automorphisms of
	colored order-$3$ hypergraphs with bounded sizes of color classes).
	
	Statements marked with an asterisk ~\apxmark\ have proofs in the attached Appendix.

	\section{Structure and decomposition of $T$-graphs}
	\label{sec:decom}
	
	In this section, we 
	give a procedure to ``extract'' a bounded number of special interval subgraphs 
	(called {\em fragments}) of a $T$-graph~$G$ in a way which is
	invariant under automorphisms and does not require a $T$-repre\-sentation on input.
	Informally, the fragments can be seen as suitable ``pieces'' of $G$ which 
	are placed on the leaves of $T$ in some representation,
	and their most important aspects are their simplicity and limited number.
	We use this extraction procedure repeatedly (and recursively) to obtain the full decomposition of a $T$-graph.

	\paragraph{\textbf{Structure of chordal graphs.}}
	
	We now give several useful terms and facts related to chordal graphs.
	A vertex $v$ of a graph $G$ is called \emph{simplicial} if its neighborhood
	corresponds to a clique of $G$.  
	It is known that every chordal graph
	contains a simplicial vertex and, by removing the simplicial vertices of a chordal graph repeatedly,
	one obtains an empty graph.
	
	A \emph{weighted clique graph} $\mathcal{C}_G$ of a graph $G$ is the graph
	whose vertices are the maximal cliques of $G$ and there is an edge
	between two vertices in $\mathcal{C}_G$ whenever the corresponding maximal
	cliques have a non-empty intersection.  The edges in $\mathcal{C}_G$ are
	weighted by the cardinality of the intersection of the corresponding~cliques.
	
	A \emph{clique tree} of $G$ is any maximum-weight spanning tree of $\mathcal{C}_G$
	which may not be unique.  An edge of $\mathcal{C}_G$ is called
	\emph{indispensable} (resp.\ \emph{unnecessary}) if it appears in
	every (resp.\ none) maximum-weight spanning tree of $\mathcal{C}_G$.
	If $G$ is chordal, every maximum-weight spanning tree $T$ of
	$\mathcal{C}_G$ is also a $T$-representation of $G$, e.g.~\cite{MATSUI20103635}.

	For a graph $G$ and two vertices $u \neq v \in V(G)$, a subset $S \subseteq
	V(G)$ is called a \emph{u-v separator} (or \emph{u-v cut}) of $G$ if $u$ and
	$v$ belong to different components of $G - S$.  When $\vert S \vert = 1$,
	then $S$ is called a \emph{cutvertex}.  $S$ is called \emph{minimal} if no
	proper subset of $S$ is a $u$-$v$ separator.  Minimal separators of a
	graph are the separators which are minimal for some pair of vertices. 
	Chordal graphs, thus $T$-graphs, have linearly many minimal vertex
	separators \cite{linearlyManySeparatorOfChordal}.
	
	A \emph{leaf clique} of a $T$-graph $G$ is a maximal clique of $G$ which can
	be a leaf of some clique tree of $G$ (informally, it can be placed on a leaf 
	of $T$ in some $T$-representation of $G$). 
	We use the following lemma in our algorithm:
	
	\begin{lemma}[Matsui et al.~\cite{MATSUI20103635}]\label{leafCondition}
		A maximal clique $C$ of a chordal graph $G$ can be a leaf of a
		clique tree if and only if $C$ satisfies (1) $C$ is incident to at most one
		indispensable edge of $\mathcal{C}_G$, and (2) $C$ is not a cutvertex in
		$\mathcal{C}_G'$ which is the subgraph of $\mathcal{C}_G$ which includes all
		edges except the unnecessary ones.
		The conditions can be checked in polynomial time.
	\end{lemma}

	\paragraph{\textbf{Decomposing $T$-graphs.}}
	
	The overall goal now is to recursively find a unique decomposition of a given $T$-graph $G$ 
	into levels such that each level consists of a bounded number of interval fragments.
	
	For an illustration, a similar decomposition can be obtained directly from a
	$T$-representation of~$G$:
	pick the interval subgraphs of $G$ which are represented exclusively on the leaf
	edges of $T$, forming the outermost level, and recursively in the same way obtain
	the next levels.
	Unfortunately, this is not a suitable solution for us, not only that we do not
	have a $T$-decomposition at hand, but mainly because we need our
	decomposition to be {\em canonical}, meaning invariant under automorphisms of the
	graph, while this depends on a particular representation.
	
	The contribution of this section is to compute such a decomposition the right canonical way.
	As sketched above, the core task is to canonically determine in the given graph $G$ one
	bounded-size collection of fragments which will form the outermost level of the
	decomposition, and then the rest of the decomposition is obtained in the same way from
	recursively computed collections of fragments in the rest of the graph, which is also a $T$-graph%
	\footnote{Since the requirement of canonicity of our collection does not
		allow us to relate this collection to a particular $T$-representation of
		$G$, we cannot say whether the rest of $G$ (after removing our collection of
		fragments) would be a $T_1$-graph for some strict subtree $T_1\subsetneq T$, 
		or only a $T$-graph again. That is why we speak about $T$-graphs for the
		same $T$ (or, we could say graphs of bounded leafage here)
		throughout the whole recursion.
		In particular, we cannot directly use this procedure to recognize $T$-graphs.
	}.
	
	\smallskip
	
	For a chordal graph $G$ and a (fixed) collection $Z_1,Z_2,\ldots,Z_s\subseteq G$
	of distinct cliques, we write $Z_i\preceq Z_j$ if there exists
	$k\in\{1,\dots,s\}\setminus\{i,j\}$ such that $Z_j$ {\em separates} $Z_i$
	from $Z_k$ in $G$ (meaning that there is no path from $Z_i\setminus Z_j$ to
	$Z_k\setminus Z_j$ in $G-Z_j$), and say that $Z_i\preceq Z_j$ is witnessed by~$Z_k$.
	Note that $\preceq$ is transitive, and hence a preorder.
	Let $Z_i\precneqq Z_j$ mean that $Z_i\preceq Z_j$ but $Z_j\not\preceq Z_i$.
	We also write $Z_i\approx Z_j$ if there exists
	$k\in\{1,\dots,s\}\setminus\{i,j\}$ such that both $Z_i\preceq Z_j$ and
	$Z_j\preceq Z_i$ hold and are witnessed by~$Z_k$.
	Note that $Z_j\approx Z_i$ is stronger than just saying `$Z_i\preceq Z_j$ and     
	$Z_j\preceq Z_i$,' and that $Z_i\cap Z_j$ then separates $Z_i\Delta Z_j$
	from~$Z_k$.
	
	\begin{replemma}{cliqdd}\apxmark\label{lem:cliqdd}
		Let $T$ be a tree with $d$ leaves, and $G$ be a $T$-graph.
		Assume that $Z_1,\ldots,Z_s\subseteq G$ are distinct cliques of $G$ such that
		one of the following holds:
		\\a) for each $1\leq i\leq s$, the set $Z_i$ is a maximal clique in $G$,
		and for any $1\leq i\not=j\leq s$, neither $Z_i\precneqq Z_j$ nor
		$Z_i\approx Z_j$ is true, or
		\\b) for each $1\leq i\leq s$, the set $Z_i$ is a minimal separator in $G$
		cutting off a component $F$ of $G-Z_i$ such that $F$ contains a simplicial
		vertex of a leaf clique of~$G$, and that $F$ is disjoint from all $Z_j$, $j\not=i$.
		\\Then $s\leq d$.
	\end{replemma}
	
	Let $Z\subseteq G$ be a minimal separator in~$G$
	and $F\subseteq G$ a connected component of $G-Z$.
	Then $Z$ is a clique since $G$ is chordal, and whole $Z$ is in the
	neighborhood of~$F$ by minimality.
	We call a {\em completion of $F$} (in implicit~$G$) the graph $F^+$ obtained
	by contracting all vertices of $G$ not in $V(F)\cup Z$ into one vertex $l$
	(the neighborhood of $l$ is thus~$Z$) and joining $l$ with a new
	leaf vertex~$l'$, called the {\em tail of~$F^+$}.
	Since $F$ determines~$Z$ in a chordal graph~$G$, the term $F^+$ is well defined.
	
	We call a collection of disjoint nonempty induced 
	subgraphs (not necessarily connected) $X_1,X_2,\ldots,X_s\subseteq G$, 
	such that there are no edges between distinct $X_i$ and $X_j$, 
	a {\em fragment collection} of $G$ of size~$s$.
	We first give our procedure for computing a fragment collection, 
	and subsequently formulate (and prove) the crucial properties
	of the computed collection and the whole decomposition.
	
	\begin{proc}\label{proc:extract} \rm
		Let $T$ be a tree with $d$ leaves and no degree-$2$ vertex.
		Assume a $T$-graph $G$ on the input.
		We compute an induced (and canonical) fragment collection 
		$X_1,X_2,\ldots,X_s\subseteq G$ of $G$ of size $0<s\leq2d$ as follows:
		\begin{enumerate}
			\item 
			List all maximal cliques in $G$ (using a simplicial decomposition)
			and compute the weighted clique graph $\mathcal{C}_G$ of $G$. 
			Compute the list $\ca L$ of all possible leaf cliques of $G$ by
			Lemma~\ref{leafCondition}; in more detail, using
			\cite[Algorithm~2]{MATSUI20103635} for computation of the indispensable edges in $\mathcal{C}_G$.
			\label{it:stepfirst}
			\item 
			For every pair $L_1,L_2\in\ca L$ such that $L_1\precneqq L_2$, remove $L_2$ from the list.
			Let $\ca L_0\subseteq\ca L$ be the resulting list of cliques,
			which is nonempty since $\precneqq$ is acyclic.
			\label{it:nonsepL}
			\item 
			Let $\ca L_1:=\big\{L\in\ca L_0: \forall L'\in\ca L_0\setminus\{L\}.~ L\not\approx L'\big\}$
			be the subcollection of cliques incomparable with others in~$\approx$.
			By Lemma~\ref{lem:cliqdd}(a) we have $|\ca L_1|\leq d$.
			If $\ca L_1\not=\emptyset$, then {\bf output} the following fragment collection of~$G$:
			for each $L\in\ca L_1$, include in it the set $F\subseteq L$ of all
			simplicial vertices of $L$ in the graph~$G$.
			\label{it:outsimplic}%
			\item \smallskip
			Now, for each $L\in\ca L_0$ we have $L'\in\ca L_0\setminus\{L\}$
			such that $L\approx L'$ (and so $L\cap L'$ is a separator in~$G$).
			For distinct $L_1,L_2\in\ca L_0$ such that $L_1\approx L_2$,
			we call a set $Z\subseteq L_1\cap L_2$ a {\em joint separator for $L_1,L_2$}
			if $Z$ separates $L_1\Delta L_2$ from $L\setminus Z$ for some (any)
			$L\in\ca L_0\setminus\{L_1,L_2\}$.
			We compute the family $\ca Z$ of all inclusion-minimal sets $Z$ which are
			joint separators for some pair $L_1\approx L_2\in\ca L_0$ as above, over all
			such pairs~$L_1,L_2$.
			This is efficient since all minimal separators in chordal graphs can be
			listed in linear time.
			Note that no set $Z\in\ca Z$ contains any simplicial vertex of~$G$, 
			and so $V(G)\not\subseteq\bigcup\ca Z$.
			\label{it:nonsepZ}
			\item 
			Let $\ca C$ be the family of the connected components of $G-\bigcup\ca Z$,
			and $\ca C_0\subseteq\ca C$ consist of such $F\in\ca C$ that $F$ is incident 
			to just one set $Z_F\in\ca Z$.
			Note that $\ca C_0\not=\emptyset$, since otherwise the incidence graph between
			$\ca C$ and $\ca Z$ would have a cycle and this would in turn contradict chordality of~$G$.
			Let $\ca Z_0:=\{Z_F\in\ca Z: F\in\ca C_0\}$.
			Moreover, by Lemma~\ref{lem:cliqdd}(b), $|\ca Z_0|\leq d$.
			\item \smallskip
			We make a collection $\ca C_0'$ from $\ca C_0$ by the following
			operation: for each $Z\in\ca Z_0$, take all $F\in\ca C_0$ such that $Z_F=Z$
			and every vertex of $F$ is adjacent to whole~$Z$, and join them into one
			graph in $\ca C_0'$ (note that there can be arbitrarily many such $F$ for
			one~$Z$).
			Remaining graphs of $\ca C_0$ stay in $\ca C_0'$ without change.
			Then, we denote by $\ca C_1\subseteq\ca C_0'$ the~subcollection of those
			$F\in\ca C_0'$ such that the completion $F^+$~of~$F$ (in $G$) is an interval graph.%
			\footnote{Informally, $F^+\in\ca C_1$ iff $F$ has an interval representation
				(on a horizontal line) to which its separator $Z_F$ can be ``attached from the left'' on the same line.}
			\label{it:Zseparators}
			\item
			If $\ca C_1\not=\emptyset$, then {\bf output} $\ca C_1$ as the fragment collection.
			(As we can show from Lemma~\ref{lem:cliqdd}(b), $|\ca C_1|\leq d+|\ca Z_0|\leq2d$.)
			\label{it:C1fragment}
			\item
			Otherwise, for each graph $F\in\ca C_0'$, we call this procedure recursively
			on the completion $F^+$ of $F$ (these calls are independent since the
			graphs in $\ca C_0'$ are pairwise disjoint).
			Among the fragments returned by this call, we keep only those which are subgraphs of~$F$.%
			\footnote{Note that, e.g., the separator and tail of $F^+$ may also be involved in a
				recursively computed fragment.}
			We {\bf output} the fragment collection formed by the
			union of kept fragments from~all~recursive~calls.
			\label{it:recurfrag}
		\end{enumerate}
	\end{proc}
	
	One call to Procedure~\ref{proc:extract} clearly takes only polynomial time
	(in some steps this depends on $G$ being chordal -- e.g., listing all cliques
	or separators).
	Since the possible recursive calls in the procedure are applied to pairwise disjoint
	parts of the graph (except the negligible completion of $F$ to $F^+$),
	the overall computation of Procedure~\ref{proc:extract} takes polynomial time
	regardless of~$d$.
	Regarding correctness, we are proving that $s\leq2d$, which is in the
	respective Steps~\ref{it:outsimplic} and~\ref{it:C1fragment} indicated as a corollary of Lemma~\ref{lem:cliqdd},
	except in the last (recursive) Step~\ref{it:recurfrag} where it
	can be derived in a similar way from Lemma~\ref{lem:cliqdd} applied to the final collection.
	We leave the remaining technical details for the attached Appendix.
	
	The last part is to prove a crucial fact that the collection $X_1,X_2,\ldots,X_s\subseteq G$
	is indeed canonical, which is precisely stated as follows:
	\begin{replemma}{reallycanonical}\apxmark\label{lem:reallycanonical}
		Let $G$ and $G'$ be isomorphic $T$-graphs.
		If Procedure~\ref{proc:extract} computes the canonical collection
		$X_1,\ldots,X_s$ for $G$ and the canonical collection
		$X_1',\ldots,X_{s'}'$ for $G'$, then $s=s'$ and there is an
		isomorphism between $G$ and $G'$ matching in some order
		$X_1,\ldots,X_s$ to $X_1',\ldots,X_{s}'$.
	\end{replemma}

	\paragraph{\textbf{Levels, attachments and terminal sets.}}
	
	Following Procedure~\ref{proc:extract}, we now show how the full
	decomposition of a $T$-graph $G$ is completed.
	
	For every fragment $X$ of the canonical collection computed by               
	Procedure~\ref{proc:extract}, we define the list of 
	{\em attachment sets} of $X$ in $G-X$ as follows.
	If $X=F$ is obtained in Step \ref{it:outsimplic}, then it has one attachment set $L\setminus F$.
	Otherwise (Steps \ref{it:Zseparators}~and~\ref{it:C1fragment}\,), 
	the attachment sets of $X=F$ are all subsets $A$ of the corresponding
	separator $Z$ (of~$F$) such that some vertex of $X$ has the neighborhood in $Z$ equal to~$A$.
	Observe that the attachment sets of $X$ are always cliques contained in the
	completion $X^+$, as defined above.
	Moreover, it is important that the attachment sets of $X$ form a chain by the
	set inclusion, since $G$ is chordal, and hence they are {\em uniquely determined}
	independently of automorphisms of~$X^+$.

	\begin{proc}\label{proc:decomprec}\rm
		Given a $T$-graph $G$, we determine a {\bf canonical decomposition}
		of $G$ recursively as follows. Start with $i=1$ and~$G_0:=G$.
		\begin{enumerate}
			\item 
			Run Procedure~\ref{proc:extract} for $G_{i-1}$, obtaining the collection $X_1,\ldots,X_s$.
			\item 
			We call the special interval subgraphs $X_1,\ldots,X_s$ {\em fragments} and their family
			$\ca X_i:=\{X_1,\ldots,X_s\}$ a {\em level} (of~number~$i$) of the constructed decomposition.
			\item 
			Let $G_{i}:=G-\big(V(X_1)\cup\ldots\cup V(X_s)\big)$.
			Mark every attachment set of each $X_j$ in $G_{i}$ as a {\em terminal set}.
			These terminal sets will be further refined when recursively decomposing
			$G_i$; namely, further constructed fragments of $G_i$ will inherit
			induced subsets of marked terminal sets as their terminal sets.
			\item 
			As long as $G_i$ is not an interval graph, repeat this from Step~1 with
			$i\leftarrow i+1$.
		\end{enumerate}
	\end{proc}
	
	Regarding this procedure, we stress that the obtained levels are numbered
	``from outside'', meaning that the first (outermost) level is of the least index.
	The rule is that fragments from lower levels have their attachment sets
	as terminal sets in higher levels.
	As it will be made precise in the next section, an isomorphism between two
	$T$-graphs can be captured by a mapping between their canonical decompositions,
	which relates pairwise isomorphic fragments and preserves the incidence (i.e., identity)
	between the attachment sets of mapped fragments and the terminal sets of
	fragments in higher levels.
	See also Figure~\ref{fig:terminalAut}.

	\section{Group-computing tools}
	\label{sec:algebraic}
	
	We first recall the notion of the \emph{automorphism group} which is closely
	related to the graph isomorphism problem.  An \emph{automorphism} is an
	isomorphism of a graph $G$ to itself, and the \emph{automorphism group} of
	$G$ is the group $Aut(G)$ of all automorphisms of $G$.  
	There exists an isomorphism
	from $G_1$ to $G_2$ if and only if the automorphism group of the disjoint union
	$H:= G_1 \uplus G_2$ contains a permutation exchanging the vertex sets of $G_1$ and $G_2$.  
	We work with automorphism groups by means of their generators;
	a subset $A$ of elements of a group $\Gamma$ is called a {\em set of
		generators} if the members of $A$ together with the operation of $\Gamma$
	can generate each element~of~$\Gamma$.
	
	There are two related classical algebraic tools which we shall use in the next section.
	The first one is an algorithm
	performing computation of a subgroup of an arbitrary group, provided that we
	can efficiently test the membership in the subgroup and the subgroup is not
	``much smaller'' than the original group:
	
	\begin{theorem}{\bf(Furst, Hopcroft and Luks
			\cite[Cor.~1]{furst})}\label{thm:furstgen}
		Let $\Pi$ be a permutation group given by its generators, 
		and $\Pi_1$ be any subgroup of $\Pi$ such that one can test in
		polynomial time whether $\pi\in\Pi_1$ for any $\pi\in\Pi$ (membership test).
		If the ratio $|\Pi|/|\Pi_1|$ is bounded by a function of a parameter $d$,
		then a set of generators of $\Pi_1$ can be computed in
		\emph{FPT}-time (with respect to~$d$).
	\end{theorem}
	
	The second tool, known as Babai's ``tower-of-groups'' procedure
	(cf.~\cite{babai-bdcm}), will not be used as a standalone statement, but as
	a mean of approaching the task of computation of the automorphism group of
	our object $H$ (e.g., graph).
	This procedure can be briefly outlined as follows; imagine an
	inclusion-ordered chain~of~groups
	$\Gamma_0\supseteq\Gamma_1\supseteq\ldots\supseteq\Gamma_{k-1}\supseteq\Gamma_{k}$
	such that
	\begin{itemize}
		\item $\Gamma_0$ is a group of some unrestricted permutations on the
		ground set of our~$H$,
		\item for each $i\in\{1,\ldots,k\}$, we ``add'' some further restriction
		(based on the structure of~$H$) which has to be satisfied by all
		permutations of $\Gamma_i$, 
		\item the restriction in the previous point is chosen such that the
		ratio $|\Gamma_{i-1}|/|\Gamma_i|$ is guaranteed to be ``small'', and
		\item in $\Gamma_k$, we get the automorphism group of our object~$H$.
	\end{itemize}
	Then Theorem~\ref{thm:furstgen} can be used to compute $\Gamma_1$ from
	$\Gamma_0$, then $\Gamma_2$ from $\Gamma_1$, and so on until we get the
	automorphism group $\Gamma_k$.

	\paragraph{\textbf{Automorphism group of a decomposition.}}
	
	Here we are going to apply the above procedure in order to compute the
	automorphism group of a special object which combines the decompositions
	(cf.~Procedure~\ref{proc:decomprec}) of given $T$-graphs $G_1$ and $G_2$, but
	abstracts from precise structure of the fragments as interval graphs.
	
	Consider canonical decompositions of the graphs $G_1$ and $G_2$, as produced
	by Procedure~\ref{proc:decomprec} in the form of level families
	$\ca X^1_1,\ldots,\ca X^1_\ell$ and
	$\ca X^2_1,\ldots,\ca X^2_{\ell'}$, respectively.
	We may assume that $\ell=\ell'$ since otherwise we immediately answer `not isomorphic'.
	A combined decomposition of $H= G_1 \uplus G_2$ hence consists of the
	levels $\ca X_i:=\ca X^1_i\cup\ca X^2_i$ for $i=1,\ldots,\ell$
	and their respective terminal sets.
	More precisely, let $\ca X:=\ca X_1\cup\ldots\cup\ca X_\ell$.
	Let $\ca A[X]$ for $X\in\ca X_k$ be the family of all terminal sets in~$X$
	(as marked by Procedure~\ref{proc:decomprec} and then restricted to~$V(X)$\,),
	and specially $\ca A^i[X]\subseteq\ca A[X]$ be those terminal sets in~$X$
	which come from attachment sets of fragments on level~$i<k$.
	Let $\ca A_k:=\bigcup_{X\in\ca X_k}\ca A[X]$ and
	$\ca A_k^i:=\bigcup_{X\in\ca X_k}\ca A^i[X]$ for $k=1,\ldots,\ell$,
	and let $\ca A:=\ca A_1\cup\ldots\cup\ca A_\ell$.
	
	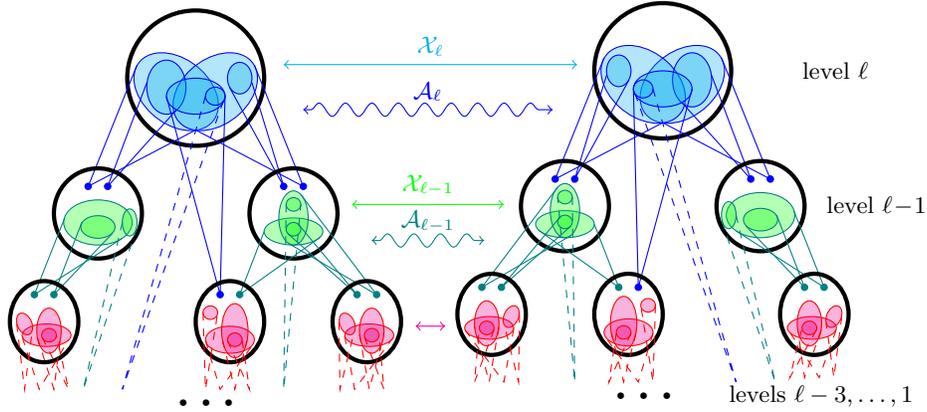
\begin{figure}[tb]
		\centering
		\begin{subfigure}[t]{0.47\linewidth}
			\centering
			\begin{tikzpicture}[xscale=0.65,yscale=0.6]
				
				\draw[<->,cyan] (-0.7,5.3) -- (5.3,5.3) node[midway,above]{$\ca X_\ell$};
				
				\draw[<->,blue,snake=snake,line before snake=1mm,line after snake=1mm] (-0.3,4.3) -- (4.8,4.3) node[midway, above]{$\ca A_\ell$};
				
				\draw[<->,green] (0.7,2.2) -- (3.8,2.2) node[midway, above]{$\ca X_{\ell-1}$};
				
				\draw[<->,teal,snake=snake,line before snake=1mm,line after snake=1mm] (1.1,1.4) -- (3.4,1.4) node[midway, above]{$\ca A_{\ell-1}$};
				
				\draw[<->,magenta] (2,-0.5) -- (2.6,-0.5) node[midway, above]{};
				
				\draw[rotate around={90:(-2.5,5)}, color=black, ultra thick, opacity=1] (-2.5,5) ellipse (1.5cm and 1.4cm); 
				
				\draw[rotate around={135:(-2.9,4.7)}, color=blue, fill=cyan, fill opacity=0.3] (-2.9,4.7) ellipse (1cm and 0.7cm);
				\draw[-,blue] (-3.7,5.29) -- (-4.7,2.6) -- (-2.5,3.85) node[midway, above]{};
				\node at (-4.7,2.6) [circle,draw, fill=blue, opacity=1, color=blue, inner sep=0.3mm] (a) {};
				
				\draw[rotate around={45:(-2.1,4.7)}, color=blue, fill=cyan, fill opacity=0.3] (-2.1,4.7) ellipse (1cm and 0.7cm); 
				\draw[-,blue] (-1.3,5.29) -- (-0.3,2.6) -- (-2.5,3.85) node[midway, above]{};
				\node at (-0.3,2.6) [circle,draw, fill=blue, opacity=1, color=blue, inner sep=0.3mm] (b) {};
				
				\draw[rotate around={90:(-3.1,4.8)}, color=blue, fill=cyan, fill opacity=0.3] (-3.1,4.8) ellipse (0.6cm and 0.4cm); 
				\draw[-,blue] (-3.5,4.9) -- (-4.3,2.6) -- (-2.89,4.3) node[midway, above]{};
				\node at (-4.3,2.6) [circle,draw, fill=blue, opacity=1, color=blue, inner sep=0.3mm] (c) {};		
				
				\draw[rotate around={0:(-2.5,4.6)}, color=blue, fill=cyan, fill opacity=0.3] (-2.5,4.6) ellipse (0.6cm and 0.4cm); 
				\draw[-,blue] (-3.1,4.6) -- (-2,0.2) -- (-1.9,4.6) node[midway, above]{};
				\node at (-2,0.2) [circle,draw, fill=blue, opacity=1, color=blue, inner sep=0.3mm] (i) {};
				
				\draw[rotate around={0:(-1.79,5.9)}, color=blue, fill=cyan, fill opacity=0.3] (-2.1,4.6) ellipse (0.2cm and 0.2cm); 
				\draw[-,blue,dashed] (-2.3,4.6) -- (-4,-1.9) -- (-1.9,4.6) node[midway, above]{};
				
				\draw[rotate around={90:(-1.6,5)}, color=blue, fill=cyan, fill opacity=0.3] (-1.6,5) ellipse (0.35cm and 0.25cm); 
				\draw[-,blue] (-1.81,4.8) -- (-0.7,2.6) -- (-1.35,4.9) node[midway, above]{};
				\node at (-0.7,2.6) [circle,draw, fill=blue, opacity=1, color=blue, inner sep=0.3mm] (c) {};

				\draw[rotate around={90:(-4.5,2)}, color=black, ultra thick, opacity=1] (-4.5,2) ellipse (1cm and 0.9cm); 
				
				\draw[rotate around={0:(-4.5,1.8)}, color=teal, fill=green, fill opacity=0.3] (-4.5,1.8) ellipse (0.7cm and 0.5cm); 	
				\draw[-,teal] (-5.18,1.9) -- (-5.8,0.2) -- (-4,1.45) node[midway, above]{};
				\node at (-5.8,0.2) [circle,draw, fill=teal, opacity=1, color=teal, inner sep=0.3mm] (d) {};
				
				\draw[rotate around={0:(-4.5,1.7)}, color=teal, fill=green, fill opacity=0.3] (-4.5,1.7) ellipse (0.35cm and 0.25cm); 	
				\draw[-,teal] (-4.83,1.8) -- (-5.4,0.2) -- (-4.2,1.56) node[midway, above]{};
				\node at (-5.4,0.2) [circle,draw, fill=teal, opacity=1, color=teal, inner sep=0.3mm] (e) {};
				
				\draw[rotate around={90:(-3.85,1.8)}, color=teal, fill=green, fill opacity=0.3] (-3.85,1.8) ellipse (0.3cm and 0.15cm); 	
				\draw[-,teal,dashed] (-4,1.8) -- (-4.8,-1.9) -- (-3.7,1.7) node[midway, above]{};

				\draw[rotate around={90:(-5.6,-0.4)}, color=black, ultra thick, opacity=1] (-5.6,-0.4) ellipse (0.9cm and 0.7cm); 
				
				\draw[rotate around={0:(-5.55,-0.7)}, color=red, fill=magenta, fill opacity=0.3] (-5.55,-0.7) ellipse (0.5cm and 0.25cm);
				\draw[-,red,dashed] (-6.045,-0.75) -- (-5.7,-1.9) -- (-5.1,-0.8) node[midway, above]{};
				
				\draw[rotate around={90:(-5.5,-0.6)}, color=red, fill=magenta, fill opacity=0.3] (-5.5,-0.6) ellipse (0.5cm and 0.25cm);
				\draw[-,red,dashed] (-5.74,-0.75) -- (-5.4,-1.9) -- (-5.27,-0.8) node[midway, above]{};
				
				\draw[rotate around={90:(-5.5,-0.7)}, color=red, fill=magenta, fill opacity=0.3] (-5.5,-0.7) ellipse (0.15cm and 0.15cm); 
				\draw[-,red,dashed] (-5.64,-0.75) -- (-5.1,-1.9) -- (-5.35,-0.7) node[midway, above]{};
				
				\draw[rotate around={110:(-6,-0.45)}, color=red, fill=magenta, fill opacity=0.3] (-6,-0.45) ellipse (0.25cm and 0.15cm); 
				\draw[-,red,dashed] (-6.16,-0.45) -- (-6,-1.9) -- (-5.84,-0.59) node[midway, above]{};
				
				\draw[rotate around={90:(-0.5,2)}, color=black, ultra thick, opacity=1] (-0.5,2) ellipse (1cm and 0.9cm); 
				
				\draw[rotate around={0:(-0.5,1.6)}, color=teal, fill=green, fill opacity=0.3] (-0.5,1.6) ellipse (0.6cm and 0.3cm); 	
				\draw[-,teal] (-1.1,1.6) -- (-1.6,0.2) -- (0.045,1.47) node[midway, above]{};
				\node at (-1.6,0.2) [circle,draw, fill=teal, opacity=1, color=teal, inner sep=0.3mm] (f) {};
				
				\draw[rotate around={90:(-0.5,1.9)}, color=teal, fill=green, fill opacity=0.3] (-0.5,1.9) ellipse (0.6cm and 0.3cm); 	
				\draw[-,teal] (-0.61,1.34) -- (1.2,0.2) -- (-0.35,2.42) node[midway, above]{};
				\node at (1.2,0.2) [circle,draw, fill=teal, opacity=1, color=teal, inner sep=0.3mm] (g) {};
				
				\draw[rotate around={90:(-0.5,1.65)}, color=teal, fill=green, fill opacity=0.3] (-0.5,1.65) ellipse (0.15cm and 0.15cm); 
				\draw[-,teal] (-0.615,1.55) -- (0.8,0.2) -- (-0.41,1.773) node[midway, above]{};
				\node at (0.8,0.2) [circle,draw, fill=teal, opacity=1, color=teal, inner sep=0.3mm] (h) {};
				
				\draw[rotate around={0:(-0.5,2.2)}, color=teal, fill=green, fill opacity=0.3] (-0.5,2.2) ellipse (0.15cm and 0.15cm); 	
				\draw[-,teal,dashed] (-0.65,2.2) -- (-0.7,-1.9) -- (-0.35,2.2) node[midway, above]{};
				
				\draw[rotate around={90:(-1.8,-0.4)}, color=black, ultra thick, opacity=1] (-1.8,-0.4) ellipse (0.9cm and 0.7cm); 
				
				\draw[rotate around={0:(-1.8,-0.8)}, color=red, fill=magenta, fill opacity=0.3] (-1.8,-0.8) ellipse (0.5cm and 0.3cm);
				\draw[-,red,dashed] (-2.3,-0.82) -- (-2,-1.9) -- (-1.325,-0.9) node[midway, above]{};
				
				\draw[rotate around={90:(-1.7,-0.6)}, color=red, fill=magenta, fill opacity=0.3] (-1.7,-0.6) ellipse (0.6cm and 0.27cm);
				\draw[-,red,dashed] (-1.95,-0.82) -- (-1.7,-1.9) -- (-1.45,-0.82) node[midway, above]{};
				
				\draw[rotate around={90:(-1.7,-0.8)}, color=red, fill=magenta, fill opacity=0.3] (-1.7,-0.8) ellipse (0.15cm and 0.15cm);
				\draw[-,red,dashed] (-1.85,-0.82) -- (-1.4,-1.9) -- (-1.55,-0.82) node[midway, above]{};
				
				\draw[rotate around={90:(-2.2,-0.2)}, color=red, fill=magenta, fill opacity=0.3] (-2.2,-0.2) ellipse (0.15cm and 0.15cm);
				\draw[-,red,dashed] (-2.35,-0.2) -- (-2.3,-1.9) -- (-2.05,-0.2) node[midway, above]{};
				
				\draw[rotate around={90:(1,-0.4)}, color=black, ultra thick, opacity=1] (1,-0.4) ellipse (0.9cm and 0.7cm); 
				
				\draw[rotate around={0:(1.05,-0.7)}, color=red, fill=magenta, fill opacity=0.3] (1.05,-0.7) ellipse (0.5cm and 0.25cm);
				\draw[-,red,dashed] (0.547,-0.72) -- (0.75,-1.9) -- (1.543,-0.75) node[midway, above]{};
				
				\draw[rotate around={90:(1.1,-0.6)}, color=red, fill=magenta, fill opacity=0.3] (1.1,-0.6) ellipse (0.5cm and 0.25cm);
				\draw[-,red,dashed] (0.85,-0.68) -- (1.05,-1.9) -- (1.352,-0.68) node[midway, above]{};
				
				\draw[rotate around={90:(1.1,-0.7)}, color=red, fill=magenta, fill opacity=0.3] (1.1,-0.7) ellipse (0.15cm and 0.15cm); 
				\draw[-,red,dashed] (0.951,-0.74) -- (1.35,-1.9) -- (1.252,-0.68) node[midway, above]{};
				
				\draw[rotate around={110:(0.6,-0.45)}, color=red, fill=magenta, fill opacity=0.3] (0.6,-0.45) ellipse (0.25cm and 0.15cm); 
				\draw[-,red,dashed] (0.43,-0.41) -- (0.45,-1.9) -- (0.735,-0.65) node[midway, above]{};
				
				\node (11) at (-2.2,-2.2) {\LARGE\bf\dots};	\label{a}
			\end{tikzpicture}
			
			\label{fig:UDVGd}
		\end{subfigure}
		~
		\begin{subfigure}[t]{0.49\linewidth}
			\centering
			\begin{tikzpicture}[xscale=-0.65,yscale=0.6]
				\small
				
				\node[label=right:level~{$\ell$}] at (-5,5) {};
				\node[label=right:level~{$\ell\!-\!1$}] at (-5.5,2) {};
				
				\draw[rotate around={90:(-2.5,5)}, color=black, ultra thick, opacity=1] (-2.5,5) ellipse (1.5cm and 1.4cm); 
				
				\draw[rotate around={135:(-2.9,4.7)}, color=blue, fill=cyan, fill opacity=0.3] (-2.9,4.7) ellipse (1cm and 0.7cm);
				\draw[-,blue] (-3.7,5.29) -- (-4.7,2.6) -- (-2.5,3.85) node[midway, above]{};
				\node at (-4.7,2.6) [circle,draw, fill=blue, opacity=1, color=blue, inner sep=0.3mm] (a) {};
				
				\draw[rotate around={45:(-2.1,4.7)}, color=blue, fill=cyan, fill opacity=0.3] (-2.1,4.7) ellipse (1cm and 0.7cm); 
				\draw[-,blue] (-1.3,5.29) -- (-0.3,2.6) -- (-2.5,3.85) node[midway, above]{};
				\node at (-0.3,2.6) [circle,draw, fill=blue, opacity=1, color=blue, inner sep=0.3mm] (b) {};
				
				\draw[rotate around={90:(-3.1,4.8)}, color=blue, fill=cyan, fill opacity=0.3] (-3.1,4.8) ellipse (0.6cm and 0.4cm); 
				\draw[-,blue] (-3.5,4.9) -- (-4.3,2.6) -- (-2.89,4.3) node[midway, above]{};
				\node at (-4.3,2.6) [circle,draw, fill=blue, opacity=1, color=blue, inner sep=0.3mm] (c) {};		
				
				\draw[rotate around={0:(-2.5,4.6)}, color=blue, fill=cyan, fill opacity=0.3] (-2.5,4.6) ellipse (0.6cm and 0.4cm); 
				\draw[-,blue] (-3.1,4.6) -- (-2,0.2) -- (-1.9,4.6) node[midway, above]{};
				\node at (-2,0.2) [circle,draw, fill=blue, opacity=1, color=blue, inner sep=0.3mm] (i) {};
				
				\draw[rotate around={0:(-1.79,5.9)}, color=blue, fill=cyan, fill opacity=0.3] (-2.1,4.6) ellipse (0.2cm and 0.2cm); 
				\draw[-,blue,dashed] (-2.3,4.6) -- (-4,-1.9) -- (-1.9,4.6) node[midway, above]{};
				
				\draw[rotate around={90:(-1.6,5)}, color=blue, fill=cyan, fill opacity=0.3] (-1.6,5) ellipse (0.35cm and 0.25cm); 
				\draw[-,blue] (-1.81,4.8) -- (-0.7,2.6) -- (-1.35,4.9) node[midway, above]{};
				\node at (-0.7,2.6) [circle,draw, fill=blue, opacity=1, color=blue, inner sep=0.3mm] (c) {};

				\draw[rotate around={90:(-4.5,2)}, color=black, ultra thick, opacity=1] (-4.5,2) ellipse (1cm and 0.9cm); 
				
				\draw[rotate around={0:(-4.5,1.8)}, color=teal, fill=green, fill opacity=0.3] (-4.5,1.8) ellipse (0.7cm and 0.5cm); 	
				\draw[-,teal] (-5.18,1.9) -- (-5.8,0.2) -- (-4,1.45) node[midway, above]{};
				\node at (-5.8,0.2) [circle,draw, fill=teal, opacity=1, color=teal, inner sep=0.3mm] (d) {};
				
				\draw[rotate around={0:(-4.5,1.7)}, color=teal, fill=green, fill opacity=0.3] (-4.5,1.7) ellipse (0.35cm and 0.25cm); 	
				\draw[-,teal] (-4.83,1.8) -- (-5.4,0.2) -- (-4.2,1.56) node[midway, above]{};
				\node at (-5.4,0.2) [circle,draw, fill=teal, opacity=1, color=teal, inner sep=0.3mm] (e) {};
				
				\draw[rotate around={90:(-3.85,1.8)}, color=teal, fill=green, fill opacity=0.3] (-3.85,1.8) ellipse (0.3cm and 0.15cm); 	
				\draw[-,teal,dashed] (-4,1.8) -- (-4.8,-1.9) -- (-3.7,1.7) node[midway, above]{};

				\draw[rotate around={90:(-5.6,-0.4)}, color=black, ultra thick, opacity=1] (-5.6,-0.4) ellipse (0.9cm and 0.7cm); 
				
				\draw[rotate around={0:(-5.55,-0.7)}, color=red, fill=magenta, fill opacity=0.3] (-5.55,-0.7) ellipse (0.5cm and 0.25cm);
				\draw[-,red,dashed] (-6.045,-0.75) -- (-5.7,-1.8) -- (-5.1,-0.8) node[midway, above]{}; 				
				\draw[rotate around={90:(-5.5,-0.6)}, color=red, fill=magenta, fill opacity=0.3] (-5.5,-0.6) ellipse (0.5cm and 0.25cm);
				\draw[-,red,dashed] (-5.74,-0.75) -- (-5.4,-1.8) -- (-5.27,-0.8) node[midway, above]{}; 				
				\draw[rotate around={90:(-5.5,-0.7)}, color=red, fill=magenta, fill opacity=0.3] (-5.5,-0.7) ellipse (0.15cm and 0.15cm); 
				\draw[-,red,dashed] (-5.64,-0.75) -- (-5.1,-1.8) -- (-5.35,-0.7) node[midway, above]{}; 				
				\draw[rotate around={110:(-6,-0.45)}, color=red, fill=magenta, fill opacity=0.3] (-6,-0.45) ellipse (0.25cm and 0.15cm); 
				\draw[-,red,dashed] (-6.16,-0.45) -- (-6,-1.8) -- (-5.84,-0.59) node[midway, above]{}; 				
				\draw[rotate around={90:(-0.5,2)}, color=black, ultra thick, opacity=1] (-0.5,2) ellipse (1cm and 0.9cm); 
				
				\draw[rotate around={0:(-0.5,1.6)}, color=teal, fill=green, fill opacity=0.3] (-0.5,1.6) ellipse (0.6cm and 0.3cm); 	
				\draw[-,teal] (-1.1,1.6) -- (-1.6,0.2) -- (0.045,1.47) node[midway, above]{};
				\node at (-1.6,0.2) [circle,draw, fill=teal, opacity=1, color=teal, inner sep=0.3mm] (f) {};
				
				\draw[rotate around={90:(-0.5,1.9)}, color=teal, fill=green, fill opacity=0.3] (-0.5,1.9) ellipse (0.6cm and 0.3cm); 	
				\draw[-,teal] (-0.61,1.34) -- (1.2,0.2) -- (-0.35,2.42) node[midway, above]{};
				\node at (1.2,0.2) [circle,draw, fill=teal, opacity=1, color=teal, inner sep=0.3mm] (g) {};
				
				\draw[rotate around={90:(-0.5,1.65)}, color=teal, fill=green, fill opacity=0.3] (-0.5,1.65) ellipse (0.15cm and 0.15cm); 
				\draw[-,teal] (-0.615,1.55) -- (0.8,0.2) -- (-0.41,1.773) node[midway, above]{};
				\node at (0.8,0.2) [circle,draw, fill=teal, opacity=1, color=teal, inner sep=0.3mm] (h) {};
				
				\draw[rotate around={0:(-0.5,2.2)}, color=teal, fill=green, fill opacity=0.3] (-0.5,2.2) ellipse (0.15cm and 0.15cm); 	
				\draw[-,teal,dashed] (-0.65,2.2) -- (-0.7,-1.9) -- (-0.35,2.2) node[midway, above]{};
				
				\draw[rotate around={90:(-1.8,-0.4)}, color=black, ultra thick, opacity=1] (-1.8,-0.4) ellipse (0.9cm and 0.7cm); 
				
				\draw[rotate around={0:(-1.8,-0.8)}, color=red, fill=magenta, fill opacity=0.3] (-1.8,-0.8) ellipse (0.5cm and 0.3cm);
				\draw[-,red,dashed] (-2.3,-0.82) -- (-2,-1.9) -- (-1.325,-0.9) node[midway, above]{};
				
				\draw[rotate around={90:(-1.7,-0.6)}, color=red, fill=magenta, fill opacity=0.3] (-1.7,-0.6) ellipse (0.6cm and 0.27cm);
				\draw[-,red,dashed] (-1.95,-0.82) -- (-1.7,-1.9) -- (-1.45,-0.82) node[midway, above]{};
				
				\draw[rotate around={90:(-1.7,-0.8)}, color=red, fill=magenta, fill opacity=0.3] (-1.7,-0.8) ellipse (0.15cm and 0.15cm);
				\draw[-,red,dashed] (-1.85,-0.82) -- (-1.4,-1.9) -- (-1.55,-0.82) node[midway, above]{};
				
				\draw[rotate around={90:(-2.2,-0.2)}, color=red, fill=magenta, fill opacity=0.3] (-2.2,-0.2) ellipse (0.15cm and 0.15cm);
				\draw[-,red,dashed] (-2.35,-0.2) -- (-2.3,-1.9) -- (-2.05,-0.2) node[midway, above]{};
				
				\draw[rotate around={90:(1,-0.4)}, color=black, ultra thick, opacity=1] (1,-0.4) ellipse (0.9cm and 0.7cm); 
				
				\draw[rotate around={0:(1.05,-0.7)}, color=red, fill=magenta, fill opacity=0.3] (1.05,-0.7) ellipse (0.5cm and 0.25cm);
				\draw[-,red,dashed] (0.547,-0.72) -- (0.75,-1.9) -- (1.543,-0.75) node[midway, above]{};
				
				\draw[rotate around={90:(1.1,-0.6)}, color=red, fill=magenta, fill opacity=0.3] (1.1,-0.6) ellipse (0.5cm and 0.25cm);
				\draw[-,red,dashed] (0.85,-0.68) -- (1.05,-1.9) -- (1.352,-0.68) node[midway, above]{};
				
				\draw[rotate around={90:(1.1,-0.7)}, color=red, fill=magenta, fill opacity=0.3] (1.1,-0.7) ellipse (0.15cm and 0.15cm); 
				\draw[-,red,dashed] (0.951,-0.74) -- (1.35,-1.9) -- (1.252,-0.68) node[midway, above]{};
				
				\draw[rotate around={110:(0.6,-0.45)}, color=red, fill=magenta, fill opacity=0.3] (0.6,-0.45) ellipse (0.25cm and 0.15cm); 
				\draw[-,red,dashed] (0.43,-0.41) -- (0.45,-1.9) -- (0.735,-0.65) node[midway, above]{};

				\node (11) at (-2.2,-2.2) {\LARGE\bf\dots};	\label{b}
				\node[label=right:levels~{$\ell-3,\ldots,1$}] at (-3.5,-2.2) {};
			\end{tikzpicture}
		\end{subfigure}
		\caption{An illustration of a (combined) canonical decomposition of the graph
			$H= G_1 \uplus G_2$ into $\ell$ levels, with the collections of fragments
			$\ca X$ (thick black circles) and of terminal sets $\ca A$ (colored
			ellipses inside them). The arrows illustrate an
			automorphism of this decomposition: straight arrows show the
			possible mapping between isomorphic fragments on the same level, as
			in \ref{it:automlev}, and wavy arrows indicate preservation of the incidence
			between attachment sets and the corresponding terminal sets, as
			stated by condition \ref{it:autombetw}.}
		\label{fig:terminalAut}
	\end{figure}

	Recall, from Section~\ref{sec:decom}, the definition of the completion
	$X^+$ of any $X\in\ca X_i$ which, in the current context, is defined with
	respect to the subgraph of $H$ induced on the union $U$ of vertex sets
	of $\ca X_{i+1}\cup\ldots\cup\ca X_\ell$ (of the higher levels from~$X$).
	This is, exactly, the completion of $X$ defined by the call to
	Procedure~\ref{proc:extract} on the level~$i$ which defined $X$ as a fragment.
	Recall also the attachment sets of $X$ which are subsets of~$U$ (in $X^+$)
	and invariant on automorphisms of~$X^+$.
	
	The {\em automorphism group of such a decomposition of~$H$} (Figure~\ref{fig:terminalAut}) acts on the
	ground set $\ca X\cup\ca A$, and consists of permutations $\varrho$ of $\ca X\cup\ca A$
	which, in particular, map $\ca X_i$ onto $\ca X_i$ and $\ca A_i$ onto $\ca A_i$
	for all $i=1,\ldots,\ell$.
	Overall, we would like the permutation $\varrho$ correspond to an actual
	automorphism of the graph~$H$, for which purpose we introduce the following definition.
	A permutation $\varrho$ of $\ca X\cup\ca A$ is an {\em automorphism of the
		decomposition $(\ca X,\ca A)$ of~$H$} if the following hold true;
	\def\Aconditions{%
		\begin{enumerate}[label={\textbf(A\arabic*)}]
			\item\label{it:automlev}
			for each $X\in\ca X_i$ where $i\in\{1,\ldots,\ell\}$, we have $\varrho(X)\in\ca X_i$,
			and there is a graph isomorphism from the completion $X^+$ to the completion $\varrho(X)^+$
			mapping the tail of $X^+$ to the tail of $\varrho(X)^+$ and the terminal sets in $\ca A^j[X]$
			to the terminal sets in $\ca A^j[\varrho(X)]$ for each $1\leq j<i$, and
			\item\label{it:autombetw}
			for every $X\in\ca X_i$ and $A\in\ca A_k^i$ where
			$i\in\{1,\ldots,\ell\}$ and $k\in\{i+1,\ldots,\ell\}$, we have that if $A$ is an
			attachment set of the fragment~$X$ (so, $A\subseteq X^+$), 
			then $\varrho(A)\subseteq\varrho(X)^+$ is the corresponding 
			attachment set of the fragment~$\varrho(X)$.
		\end{enumerate}
	}\Aconditions
	
	Notice the role of the last two conditions.
	While \ref{it:automlev} speaks about consistency of $\varrho$ with the
	actual graph $H$ on the same level, \ref{it:autombetw} on the other hand ensures
	consistency ``between the levels''.
	Right from this definition we get:
	
	\begin{repproposition}{autconsistent}\apxmark\label{prop:autconsistent}
		Let $H= G_1 \uplus G_2$ and its canonical decomposition (Procedure~\ref{proc:decomprec})
		formed by families $\ca X$ and $\ca A$ be as above.
		A permutation $\varrho$ of $\ca X\cup\ca A$ is an automorphism of this
		decomposition $(\ca X,\ca A)$ of~$H$, if and only if there exists a graph automorphism of $H$ which acts
		on $\ca X$ and on $\ca A$ identically to~$\varrho$.
	\end{repproposition}

	\section{Main algorithm}
	\label{sec:mainalg}
	
	We are now ready to present our main result which gives an \emph{FPT}-time
	algorithm for isomorphism of $T$-graphs (without need for a given decomposition).
	The algorithm is based on Proposition~\ref{prop:autconsistent}, and so on
	efficient checking of the conditions \ref{it:automlev} and
	\ref{it:autombetw} in the combined decomposition of two graphs.
	Stated precisely:
	
	\begin{theorem}\label{thm:main}
		For a fixed tree $T$, there is an \emph{FPT}-time algorithm that, given graphs
		$G_1$ and $G_2$, correctly decides whether $G_1\simeq G_2$, or correctly answers that
		one or both of $G_1$ and $G_2$ are not $T$-graphs.%
		\footnote{The latter outcome (`not a $T$-graph') happens when some of the
			assertions assuming a $T$-graph in Procedure~\ref{proc:extract} fails.}
	\end{theorem}
	
	We first state a reformulation of it as a direct corollary.
	
	\begin{repcorollary}{cormain}\apxmark\label{cor:main}
		The graph isomorphism problem of chordal graphs $G_1$ and $G_2$ is in \emph{FPT}
		parameterized by the leafage of $G_1$ and $G_2$.
	\end{repcorollary}
	
	Theorem~\ref{thm:main} now follows using Procedure~\ref{proc:decomprec},
	basic knowledge of automorphism groups and Proposition~\ref{prop:autconsistent},
	and the following refined statement.
	
	\begin{reptheorem}{mainautom}\apxmark\label{thm:mainautom}
		Assume two $T$-graphs $G_1$ and $G_2$, and their combined canonical decomposition
		(Procedure~\ref{proc:decomprec}) formed by families $\ca X$ and $\ca A$
		in $\ell$ levels, as in Section~\ref{sec:algebraic}.
		Let $s=\max_{1\leq i\leq\ell}|\ca X_i|$ be the maximum size of a
		level, and $t$ be an upper bound on the maximum antichain
		size among the terminal set families~$\ca A[X]$ over each~$X\in\ca X$.
		Then the automorphism group of the decomposition, defined by \ref{it:automlev} and \ref{it:autombetw}
		above, can be computed in \emph{FPT}-time with the parameter~$s+t$.
	\end{reptheorem}
	Notice that, in our situation, the parameter $s+t$ indeed is bounded in
	terms of $|T|$; we have $s\leq2d$ and $t\leq d$ directly from
	the arguments in Lemma~\ref{lem:cliqdd} and Procedure~\ref{proc:extract}.
	Due to space limits, we give only a sketch of the proof here.
	
	\begin{proof}[sketch]
		First, we outline that the condition \ref{it:automlev} can be dealt with (in Step~\ref{it:step1b} below)
		efficiently w.r.t.\ the parameter~$t$:
		the arguments combine the known and nice description of interval
		graphs via so-called PQ-trees \cite{recogIntervalLinear,AutMPQtrees},
		with an \emph{FPT}-time algorithm~\cite{aaolu2019isomorphism} for
		the automorphism group of set families with bounded-size antichain
		(where the latter assumption is crucial for this to work).
		
		Using the previous, we prove the rest as a commented algorithm outline:
		\begin{enumerate}\parskip3pt
			\item For every level $k\in\{1,\ldots,\ell\}$ of the decomposition of $H=G_1 \uplus G_2$ we
			compute the following permutation group $\Lambda_k$ acting on $\ca X_k\cup\ca A_k$.
			\label{it:step1b}
			\begin{enumerate}[label=\alph*)]
				\item We partition $\ca X_k$ into classes according to the isomorphism
				condition \ref{it:automlev}; i.e., $X_1,X_2\in\ca X_k$ fall into the
				same class iff there is a graph isomorphism from $X_1^+$ to $X_2^+$
				preserving the tail and bijectively mapping $\ca A^i[X_1]$ to $\ca A^i[X_2]$ for all $1\leq i<k$.
				We add the bounded-order symmetric subgroup on each such class
				of $\ca X_k$~to~$\Lambda_k$.
				\item Now, for every permutation $\varrho\in\Lambda_k$ of $\ca X_k$ and
				all $X\in\ca X_k$, and for any chosen isomorphism $\iota_X: X^+\to\varrho(X)^+$
				conforming to \ref{it:automlev}, we add to $\Lambda_k$ the
				permutation of $\ca A_k$ naturally composed of partial mappings of the
				terminal sets induced by the isomorphisms $\iota_X$ over $X\in\ca X_k$.
				\item For every $X\in\ca X_k$, we compute generators of the
				automorphism subgroup of $X^+$ (preserving the tail) which maps $\ca A^i[X]$ to $\ca A^i[X]$
				for every $1\leq i<k$, and we add to~$\Lambda_k$ the action of each such generator on 
				$\ca A[X]\subseteq\ca A_k$ (as a new generator of~$\Lambda_k$).
				This part together with (a), as outlined above,
				is a nontrivial algorithmic task~\cite{A2021autmarked} and we provide further details in the attached Appendix.
			\end{enumerate}
			
			\item We let $\Gamma_0=\Lambda_1\times\ldots\times\Lambda_\ell$ be the direct
			product of the previous subgroups.
			Notice that $\Gamma_0$ is formed by the permutations conforming to condition \ref{it:automlev}.
			
			\item Finally, we apply Babai's tower-of-groups procedure \cite{babai-bdcm} to $\Gamma_0$
			in order to compute the desired automorphism group of the decomposition.
			We loop over all pairs $1\leq i<j\leq \ell$ of levels
			and over all cardinalities $r$ of terminal sets in $\ca A_j$,
			which is $\ca O(n^3)$ iterations, and in iteration $k=1,2\ldots$ compute:
			\label{it:Btower}
			\begin{itemize}
				\item[*] $\Gamma_{k}\subseteq\Gamma_{k-1}$ consisting of exactly those
				automorphisms which conform to the condition \ref{it:autombetw} for every
				component $X\in\ca X_i$ and every terminal set $A\in\ca A^i_j$ such that $|A|=r$.
				Then $\Gamma_k$ forms a subgroup of $\Gamma_{k-1}$ (i.e., closed on a composition) thanks to
				the condition \ref{it:automlev} being true in $\Gamma_{k-1}$,
				and so we can compute $\Gamma_k$ using Theorem~\ref{thm:furstgen}.
			\end{itemize}
			
			\item We output the group $\Gamma_m$ of the last iteration $k=m$ of
			Step~\ref{it:Btower} as the result.
		\end{enumerate}
		
		Correctness of the outcome of this algorithm is self-explanatory from the
		outline; $\Gamma_m$ satisfies \ref{it:automlev} and \ref{it:autombetw} for
		all possible choice of $X$ and~$A$.
		
		We finish with a brief argument of why the computation in Step~\ref{it:Btower} via
		Theorem~\ref{thm:furstgen} is indeed efficient.
		Observe that for all $i,j$,\, $|\ca X_i|\leq s$ and the number
		of $A\in\ca A^i_j$ such that $|A|=r$ is at most~$st$.
		By standard algebraic means (counting cosets of $\Gamma_{k}$ in $\Gamma_{k-1}$),
		we get that $|\Gamma_{k-1}|/|\Gamma_{k}|$ is bounded from above by the order
		of the subgroup ``induced'' on $\ca X_i$ times the order of the subgroup
		on considered sets $A$ of cardinality $r$.
		The latter number is at most $s!\cdot(st)!$ regardless of $\Gamma_{k-1}$,
		and hence bounded in the parameter.
		\qed \end{proof}

	\section{Conclusions}
	\label{sec:conclu}
	
	We have provided an \emph{FPT}-time algorithm to solve the isomorphism
	problem for $T$-graphs with a fixed parameter $|T|$ and for chordal graphs
	of bounded leafage.
	There seems to be little hope to further extend this result for more
	general classes of chordal graphs since already for split graphs of
	unbounded leafage the isomorphism problem is GI-complete.
	Though, we may combine our result with that of
	Krawczyk~\cite{DBLP:journals/corr/abs-1904-04501} for circular-arc graphs
	isomorphism to possibly tackle the case of $H$-graphs for which $H$ contains
	exactly one cycle.
	%
	
	On the other hand, an open question remains whether a similar
	decomposition technique as that in Section~\ref{sec:decom} can be
	used to solve the recognition problem of $T$-graphs in \emph{FPT}-time,
	since the currently best algorithm \cite{zemanWG} works in \emph{XP}-time.

	%
	%
	%
	\bibliographystyle{splncs04}
	\bibliography{Union-bibliography}

\begin{thebibliography}{10}
\providecommand{\url}[1]{\texttt{#1}}
\providecommand{\urlprefix}{URL }
\providecommand{\doi}[1]{https://doi.org/#1}

\bibitem{AHU}
Aho, A.V., Hopcroft, J.E., Ullman, J.D.: The Design and Analysis of Computer
  Algorithms. Addison-Wesley (1974)

\bibitem{DBLP:journals/corr/abs-2107-10689}
Arvind, V., Nedela, R., Ponomarenko, I., Zeman, P.: Testing isomorphism of
  chordal graphs of bounded leafage is fixed-parameter tractable. CoRR
  \textbf{abs/2107.10689} (2021), \url{https://arxiv.org/abs/2107.10689}

\bibitem{A2021autmarked}
A\u{g}ao\u{g}lu, D., Hlin\v{e}n{\'{y}}, P.: Automorphism group of marked
  interval graphs in {FPT}. CoRR  \textbf{abs/2202.12664} (2022),
  \url{https://arxiv.org/abs/2202.12664}

\bibitem{aaolu2019isomorphism}
A\u{g}ao\u{g}lu, D., Hlin\v{e}n{\'{y}}, P.: Isomorphism problem for
  {S{\_}d}-graphs. In: Esparza, J., Kr{\'{a}}l', D. (eds.) 45th International
  Symposium on Mathematical Foundations of Computer Science, {MFCS} 2020,
  August 24-28, 2020, Prague, Czech Republic. LIPIcs, vol.~170, pp. 4:1--4:14.
  Schloss Dagstuhl - Leibniz-Zentrum f{\"{u}}r Informatik (2020).
  \doi{10.4230/LIPIcs.MFCS.2020.4},
  \url{https://doi.org/10.4230/LIPIcs.MFCS.2020.4}

\bibitem{SdT-graphs2021efficient}
A\u{g}ao\u{g}lu, D., Hlin\v{e}n{\'{y}}, P.: Efficient isomorphism for
  {$S_d$}-graphs and {$T$}-graphs. CoRR  \textbf{abs/1907.01495} (2021)

\bibitem{babai-bdcm}
Babai, L.: Monte {C}arlo algorithms in graph isomorphism testing.
  Tech.~Rep.~79-10, Universit\'e de Montr\'eal  (1979), 42 pages

\bibitem{degree}
Babai, L., Luks, E.M.: Canonical labeling of graphs. In: Johnson, D.S., Fagin,
  R., Fredman, M.L., Harel, D., Karp, R.M., Lynch, N.A., Papadimitriou, C.H.,
  Rivest, R.L., Ruzzo, W.L., Seiferas, J.I. (eds.) Proceedings of the 15th
  Annual {ACM} Symposium on Theory of Computing, 25-27 April, 1983, Boston,
  Massachusetts, {USA}. pp. 171--183. {ACM} (1983).
  \doi{10.1145/800061.808746}, \url{https://doi.org/10.1145/800061.808746}

\bibitem{biro}
Biró, M., Hujter, M., Tuza, Z.: Precoloring extension. i. interval graphs.
  Discrete Mathematics \textbf{100} pp. 267--279 (1992)

\bibitem{recogIntervalLinear}
Booth, K.S., Lueker, G.S.: Testing for the consecutive ones property, interval
  graphs, and graph planarity using {PQ}-tree algorithms. J. Comput. Syst. Sci.
   \textbf{13}(3),  335--379 (1976). \doi{10.1016/S0022-0000(76)80045-1},
  \url{https://doi.org/10.1016/S0022-0000(76)80045-1}

\bibitem{treedepth}
Bouland, A., Dawar, A., Kopczy\'{n}ski, E.: On tractable parameterizations of
  graph isomorphism. In: Thilikos, D.M., Woeginger, G.J. (eds.) Parameterized
  and Exact Computation - 7th International Symposium, {IPEC} 2012, Ljubljana,
  Slovenia, September 12-14, 2012. Proceedings. Lecture Notes in Computer
  Science, vol.~7535, pp. 218--230. Springer (2012).
  \doi{10.1007/978-3-642-33293-7\_21},
  \url{https://doi.org/10.1007/978-3-642-33293-7\_21}

\bibitem{zemanWG}
Chaplick, S., T\"{o}pfer, M., Voborn{\'{\i}}k, J., Zeman, P.: On
  {H}-topological intersection graphs. In: Bodlaender, H.L., Woeginger, G.J.
  (eds.) Graph-Theoretic Concepts in Computer Science - 43rd International
  Workshop, {WG} 2017, Eindhoven, The Netherlands, June 21-23, 2017, Revised
  Selected Papers. Lecture Notes in Computer Science, vol. 10520, pp. 167--179.
  Springer (2017). \doi{10.1007/978-3-319-68705-6\_13},
  \url{https://doi.org/10.1007/978-3-319-68705-6\_13}

\bibitem{DBLP:journals/networks/Colbourn81}
Colbourn, C.J.: On testing isomorphism of permutation graphs. Networks
  \textbf{11}(1),  13--21 (1981). \doi{10.1002/net.3230110103},
  \url{https://doi.org/10.1002/net.3230110103}

\bibitem{AutMPQtrees}
Colbourn, C.J., Booth, K.S.: Linear time automorphism algorithms for trees,
  interval graphs, and planar graphs. {SIAM} J. Comput.  \textbf{10}(1),
  203--225 (1981). \doi{10.1137/0210015}, \url{https://doi.org/10.1137/0210015}

\bibitem{hellyCARClinearISO}
Curtis, A.R., Lin, M.C., McConnell, R.M., Nussbaum, Y., Soulignac, F.J.,
  Spinrad, J.P., Szwarcfiter, J.L.: Isomorphism of graph classes related to the
  circular-ones property. Discret. Math. Theor. Comput. Sci.  \textbf{15}(1),
  157--182 (2013), \url{http://dmtcs.episciences.org/625}

\bibitem{furst}
Furst, M.L., Hopcroft, J.E., Luks, E.M.: Polynomial-time algorithms for
  permutation groups. In: 21st Annual Symposium on Foundations of Computer
  Science, Syracuse, New York, USA, 13-15 October 1980. pp. 36--41. {IEEE}
  Computer Society (1980). \doi{10.1109/SFCS.1980.34},
  \url{https://doi.org/10.1109/SFCS.1980.34}

\bibitem{chordalityInters}
Gavril, F.: The intersection graphs of subtrees in trees are exactly the
  chordal graphs. Journal of Combinatorial Theory, Series B  \textbf{16}(1),
  47--56 (1974). \doi{https://doi.org/10.1016/0095-8956(74)90094-X},
  \url{https://www.sciencedirect.com/science/article/pii/009589567490094X}

\bibitem{linearlyManySeparatorOfChordal}
Golumbic, M.C.: Algorithmic Graph Theory and Perfect Graphs (Annals of Discrete
  Mathematics, Vol 57). North-Holland Publishing Co., NLD (2004)

\bibitem{planarLinear}
Hopcroft, J.E., Wong, J.K.: Linear time algorithm for isomorphism of planar
  graphs (preliminary report). In: Constable, R.L., Ritchie, R.W., Carlyle,
  J.W., Harrison, M.A. (eds.) Proceedings of the 6th Annual {ACM} Symposium on
  Theory of Computing, April 30 - May 2, 1974, Seattle, Washington, {USA}. pp.
  172--184. {ACM} (1974). \doi{10.1145/800119.803896},
  \url{https://doi.org/10.1145/800119.803896}

\bibitem{KLAVIK201585}
Klav{\'{\i}}k, P., Kratochv{\'{\i}}l, J., Otachi, Y., Saitoh, T.: Extending
  partial representations of subclasses of chordal graphs. Theor. Comput. Sci.
  \textbf{576},  85--101 (2015). \doi{10.1016/j.tcs.2015.02.007},
  \url{https://doi.org/10.1016/j.tcs.2015.02.007}

\bibitem{DBLP:journals/corr/abs-1904-04501}
Krawczyk, T.: Testing isomorphism of circular-arc graphs - {H}su's approach
  revisited. CoRR  \textbf{abs/1904.04501} (2019),
  \url{http://arxiv.org/abs/1904.04501}

\bibitem{treewidth}
Lokshtanov, D., Pilipczuk, M., Pilipczuk, M., Saurabh, S.: Fixed-parameter
  tractable canonization and isomorphism test for graphs of bounded treewidth.
  {SIAM} J. Comput.  \textbf{46}(1),  161--189 (2017). \doi{10.1137/140999980},
  \url{https://doi.org/10.1137/140999980}

\bibitem{MATSUI20103635}
Matsui, Y., Uehara, R., Uno, T.: Enumeration of the perfect sequences of a
  chordal graph. Theor. Comput. Sci.  \textbf{411}(40-42),  3635--3641 (2010).
  \doi{10.1016/j.tcs.2010.06.007},
  \url{https://doi.org/10.1016/j.tcs.2010.06.007}

\bibitem{genus}
Miller, G.L.: Isomorphism testing for graphs of bounded genus. In: Miller,
  R.E., Ginsburg, S., Burkhard, W.A., Lipton, R.J. (eds.) Proceedings of the
  12th Annual {ACM} Symposium on Theory of Computing, April 28-30, 1980, Los
  Angeles, California, {USA}. pp. 225--235. {ACM} (1980).
  \doi{10.1145/800141.804670}, \url{https://doi.org/10.1145/800141.804670}

\bibitem{recogChordaLinear}
Rose, D.J., Tarjan, R.E., Lueker, G.S.: Algorithmic aspects of vertex
  elimination on graphs. {SIAM} J. Comput.  \textbf{5}(2),  266--283 (1976).
  \doi{10.1137/0205021}, \url{https://doi.org/10.1137/0205021}

\bibitem{isoChordalGIComp}
Zemlyachenko, V.N., Korneenko, N.M., Tyshkevich, R.I.: Graph isomorphism
  problem. J. of Soviet Mathematics  \textbf{29},  1426--1481 (1985)

\end{thebibliography}

	\newpage
	
	\appendix
	
	\section{Additions to Section~\ref{sec:decom}}
	
	The task here is to comment and add missing arguments to the computation of
	a canonical decomposition of a $T$-graph.
	We start with an illustration of a $T$-representation of a graph in Figure~\ref{fig:TGraphlay3}.
	
	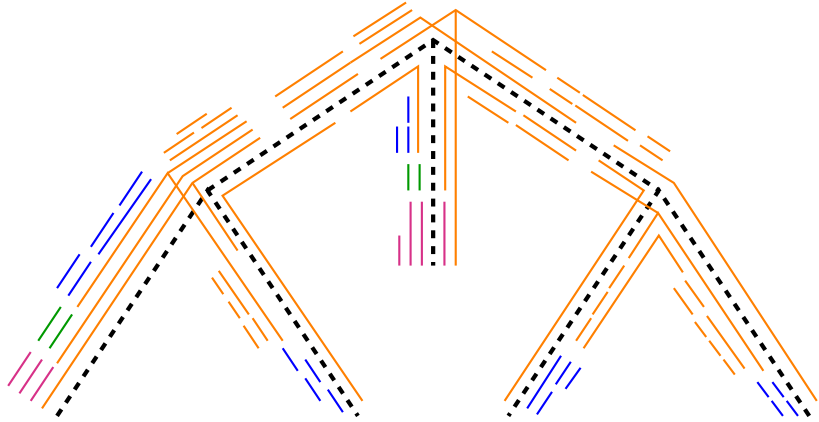
\begin{figure}[htbp]
		\centering
		\begin{tikzpicture}[scale=1.0]
			
			\draw[black, dashed, ultra thick] (-3,8) -- (0,10) -- (3,8);	
			\draw[black, dashed, ultra thick] (0,10) -- (0,7);	
			\draw[black, dashed, ultra thick] (-5,5) -- (-3,8) -- (-1,5);
			\draw[black, dashed, ultra thick] (5,5) -- (3,8) -- (1,5);
			
			\draw[orange, thick] (-5.2,5.1) -- (-3.2,8.1) -- (-2.3,8.7);
			\draw[orange, thick] (-3.2,8.1) -- (-2.6,7.2);
			
			\draw[orange, thick] (-3.5,8.4) -- (-2.6,9);
			\draw[orange, thick] (-3.58,8.5) -- (-3.18,8.77);
			\draw[orange, thick] (-3.08,8.83) -- (-2.68,9.1);
			\draw[orange, thick] (-3.41,8.75) -- (-3.01,9.02);
			
			\draw[magenta!90!black, thick] (-5.35,5.2) -- (-5.05,5.65);
			\draw[magenta!90!black, thick] (-5.5,5.3) -- (-5.2,5.75);
			\draw[magenta!90!black, thick] (-5.65,5.4) -- (-5.35,5.85);
			
			\draw[orange, thick] (-5,5.7) -- (-3.33,8.19)-- (-2.4,8.81);
			
			\draw[green!60!black, thick] (-5.1,5.9) -- (-4.8,6.35);
			\draw[green!60!black, thick] (-5.25,6) -- (-4.95,6.45);
			
			\draw[orange, thick] (-4.73,6.45) -- (-3.53,8.23) -- (-2,6);
			\draw[orange, thick] (-3.53,8.23) -- (-2.48,8.92);
			
			\draw[blue, thick] (-4.85,6.6) -- (-4.55,7.05);
			\draw[blue, thick] (-5,6.7) -- (-4.7,7.15);
			
			\draw[blue, thick] (-4.55,7.25) -- (-4.25,7.7);
			\draw[blue, thick] (-4.17,7.8) -- (-3.87,8.25);
			\draw[blue, thick] (-4.45,7.15) -- (-3.75,8.15);
			
			\draw[orange, thick] (-1.3,8.9) -- (-2.8,7.93) -- (-0.94,5.2);
			
			\draw[blue, thick] (-1.4,5.35) -- (-1.2,5.05);
			\draw[blue, thick] (-1.7,5.5) -- (-1.5,5.2);
			\draw[blue, thick] (-1.7,5.75) -- (-1.5,5.45);
			\draw[blue, thick] (-2,5.9) -- (-1.8,5.6);
			
			\draw[orange, thick] (-2.4,6.3) -- (-2.2,6);
			\draw[orange, thick] (-2.52,6.2) -- (-2.32,5.9);
			
			\draw[orange, thick] (-2.85,6.93) -- (-2.45,6.35);
			\draw[orange, thick] (-2.72,6.52) -- (-2.57,6.3);
			\draw[orange, thick] (-2.94,6.84) -- (-2.79,6.62);
			
			\draw[orange, thick] (-0.7,9.75) -- (0.3,10.4) -- (1.5,9.6);
			\draw[orange, thick] (0.3,10.4) -- (0.3,7);
			
			\draw[orange, thick] (-0.98,9.95) -- (-0.28,10.4);
			\draw[orange, thick] (-1.06,10.05) -- (-0.36,10.5);
			
			\draw[magenta!90!black, thick] (0.15,7.85) -- (0.15,7);
			\draw[magenta!90!black, thick] (-0.15,7.85) -- (-0.15,7);
			\draw[magenta!90!black, thick] (-0.3,7.85) -- (-0.3,7);
			\draw[magenta!90!black, thick] (-0.45,7.4) -- (-0.45,7);
			
			\draw[green!60!black, thick] (-0.18,8.35) -- (-0.18,8);
			\draw[green!60!black, thick] (-0.33,8.35) -- (-0.33,8);
			
			\draw[blue, thick] (-0.33,8.85) -- (-0.33,8.5);
			\draw[blue, thick] (-0.48,8.85) -- (-0.48,8.5);
			\draw[blue, thick] (-0.33,9.25) -- (-0.33,8.9);
			
			\draw[orange, thick] (-1.1,9.7) -- (-0.17,10.3) -- (1.9,8.9);
			
			\draw[orange, thick] (1.2,9) -- (1.9,8.53);
			\draw[orange, thick] (1.1,8.87) -- (1.8,8.4);
			
			\draw[orange, thick] (1.89,9.13) -- (2.59,8.66);
			\draw[orange, thick] (1.99,9.27) -- (2.69,8.8);
			
			\draw[orange, thick] (1.55,9.35) -- (1.85,9.15);
			\draw[orange, thick] (1.65,9.5) -- (1.95,9.3);
			
			\draw[orange, thick] (2.75,8.6) -- (3.05,8.4);
			\draw[orange, thick] (2.85,8.72) -- (3.15,8.52);
			
			\draw[orange, thick] (-1.1,9.05) -- (-0.2,9.65) -- (-0.2,8.5);
			
			\draw[orange, thick] (-1.9,8.95) -- (-1,9.55);
			\draw[orange, thick] (-2,9.1) -- (-1.1,9.7);
			\draw[orange, thick] (-2.1,9.25) -- (-1.2,9.85);
			
			\draw[orange, thick] (1.1,9.05) -- (0.16,9.65) -- (0.16,8);
			
			\draw[orange, thick] (1.32,9.5) -- (0.78,9.85);
			\draw[orange, thick] (1,8.9) -- (0.46,9.25);
			
			\draw[orange,thick] (0.95,5.2) -- (2.8,8) -- (2.2,8.4);
			
			\draw[blue, thick] (1.58,5.32) -- (1.38,5.02);
			\draw[blue, thick] (1.7,5.8) -- (1.25,5.1);
			\draw[blue, thick] (1.83,5.72) -- (1.63,5.42);
			\draw[blue, thick] (1.96,5.64) -- (1.76,5.36);
			
			\draw[orange, thick] (2.1,6.3) -- (1.9,6);
			\draw[orange, thick] (2.32,6.65) -- (2.12,6.35);
			\draw[orange, thick] (2.54,7) -- (2.34,6.7);
			
			\draw[orange, thick] (2,5.9) -- (3,7.4) -- (3.32,6.91);
			
			\draw[orange, thick] (4.39,5.6) -- (2.99,7.7) -- (2.6,7.1);
			\draw[orange, thick] (2.1,8.25) -- (2.99,7.7);
			
			\draw[orange, thick] (5.1,5.2) -- (3.2,8.1) -- (2.4,8.6);
			
			\draw[blue, thick] (4.85,5) -- (4.7,5.2);
			\draw[blue, thick] (4.66,5.25) -- (4.51,5.45);
			\draw[blue, thick] (4.65,5) -- (4.5,5.2);
			\draw[blue, thick] (4.46,5.25) -- (4.31,5.45);
			
			\draw[orange, thick] (4.01,5.55) -- (3.86,5.75);
			\draw[orange, thick] (3.82,5.8) -- (3.67,6);
			\draw[orange, thick] (3.63,6.05) -- (3.48,6.25);
			\draw[orange, thick] (4.14,5.63) -- (3.64,6.33);
			
			\draw[orange, thick] (3.6,6.45) -- (3.35,6.8);
			\draw[orange,  thick] (3.45,6.35) -- (3.2,6.7);
			\label{c}
			
		\end{tikzpicture}
		\caption{A $T$-graph $G$ shown by its $T$-representation (where $T$ is drawn in thick dashed black lines,
			and the subtrees shown in colored solid lines represent the vertices of~$G$).
			The picture also illustrates the first three (outermost) levels -- in order magenta, green, blue,
			of the canonical decomposition of $G$ obtained using Procedure~\ref{proc:extract}.
			The rest of $G$ is undistinguished in orange color.
			Notice that the magenta and green levels have been obtained from Step~\ref{it:outsimplic} of the procedure,
			while the blue level has come from Step~\ref{it:C1fragment}.	}
		\label{fig:TGraphlay3}
	\end{figure}

	The following observation will be useful in the coming proofs.
	\begin{itemize}
		\item[(*)] If $T$ is a tree with $d$ leaves, and we mark $d+1$ vertices of $T$,
		then some path in $T$ contains $3$ marked vertices.
	\end{itemize}
	This can be proved by successively removing unmarked leaves from $T$,
	until all $\leq d$ leaves of the remaining subtree of $T$ are marked.
	However, then also some internal vertex is marked, and it can be prolonged
	into a path ending in two leaves which are marked, too.
	
	\repeatlemma{cliqdd}
	
	\begin{proof}
		Consider a $T$-representation of the graph $G$ -- as the intersection graph
		of subtrees of a subdivision $T'$ of~$T$.
		Then every clique $X$ of $G$ must be represented such that the representatives 
		of all vertices of $X$ intersect in (at least) one common node $v[X]\in V(T')$.
		
		a) If $s>d$, then by observation (*) there exist $1\leq i<j<k\leq s$ such that the three
		(distinct) nodes $v[Z_i]$, $v[Z_j]$ and $v[Z_k]$ of $T'$ lie on one path in this order.
		Hence $Z_j$ separates $Z_i$ from $Z_k$ in~$G$ by basic properties of a $T$-representation
		and maximality of the clique~$Z_j$. 
		If $Z_i$ separates $Z_j$ from $Z_k$, too, then $Z_i\approx Z_j$. 
		Otherwise, we have $Z_i\precneqq Z_j$, and both cases lead to a contradiction.
		
		b) For $1\leq i\leq s$ and some component $F$ of $G-Z_i$, let $X_i$ be a leaf clique 
		of $G$ having a simplicial vertex in~$F$.
		As in a), assuming $s>d$, we find nodes $v[X_i]$, $v[X_j]$ and
		$v[X_k]$ of $T'$ (of three distinct maximal leaf cliques $X_i,X_j,X_k$ of~$G$)
		that lie on one path in this order.
		However, from the definition of a $T$-representation (note that the node
		$v[X_j]$ disconnects $v[X_i]$ from $v[X_k]$ in $T'$), we get that
		$Z_i\cap Z_k\subseteq Z_j$.
		From the assumption of minimality of our separators we conclude that $Z_i=Z_j=Z_k$.
		\qed \end{proof}

	\paragraph{\bf{Additions to Procedure~\ref{proc:extract}.}}
	
	We further comment on necessity of considering the joint separators in
	Step~\ref{it:nonsepZ} and subsequent steps
	of the procedure, which is the more complicated side of the procedure.
	This is, though, unavoidable in cases that $G$ contains arbitrarily many
	leaf cliques which are then mutually comparable in $\approx$.
	Such a situation is illustrated in Figure~\ref{fig:TGraphx}.
	
	We also add a bit of explanation to Step~\ref{it:C1fragment}: $|\ca C_1|\leq d+|\ca Z_0|$.
	The part $\leq|\ca Z_0|$ applies to those joined members of $\ca C_1$ which are
	fully adjacent to whole $Z\in\ca Z_0$, i.e., as $|\ca C_1\setminus\ca C_0|\leq|\ca Z_0|$.
	The part $\leq d$ applies to the remaining members of $\ca C_1$, precisely,
	as $|\ca C_1\cap\ca C_0|\leq|\ca Z_0|$ which is an application of Lemma~\ref{lem:cliqdd}(b).
	An analogous reasoning applies to the fragments collected from recursive
	calls in Step~\ref{it:recurfrag}; these satisfy the same properties with respect to the whole
	$T$-representation of $G$ as fragments directly obtained in Step~\ref{it:C1fragment}.
	\qed

	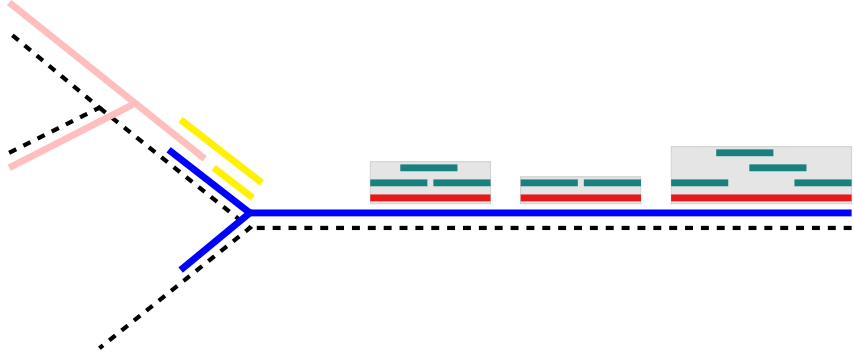
\begin{figure}[tbp]
		\centering
		\begin{tikzpicture}[xscale=-0.8,yscale=0.8]
			\draw[black, dashed, ultra thick] (-6.75,-4.5) -- (3.25,-4.5);
			\draw[black, dashed, ultra thick] (5.75,-2.5) -- (3.25,-4.5) -- (5.75,-6.5);
			\draw[black, dashed, ultra thick] (7.25,-3.25) -- (5.75,-2.5) -- (7.25,-1.25);
			
			\draw[blue, line width=2.6] (-6.75,-4.25) -- (3.25,-4.25) -- (4.6,-3.2); 
			\draw[blue, line width=2.6]  (3.25,-4.25) -- (4.4,-5.2); 
			
			\draw[yellow, line width=2.6]  (3.05,-3.75) -- (4.4,-2.7); 
			\draw[yellow, line width=2.6]  (3.2,-4) -- (3.85,-3.5); 
			
			\draw[pink, line width=2.6] (4,-3.35) -- (7.25,-0.75); 
			\draw[pink, line width=2.6] (5.2,-2.45) -- (7.25,-3.5); 
			
			\draw[red, line width=2.6] (-6.75,-4) -- (-3.75,-4); 
			\draw[teal, line width =2.6] (-6.75,-3.75) -- (-5.8,-3.75); 
			\draw[teal, line width =2.6] (-6,-3.5) -- (-5.05,-3.5); 
			\draw[teal, line width =2.6] (-5.45,-3.25) -- (-4.5,-3.25); 
			\draw[teal, line width =2.6] (-4.7,-3.75) -- (-3.75,-3.75); 
			
			\draw [draw=gray, fill=gray, opacity=0.2] (-6.75,-4.1) rectangle (-3.75,-3.15);
			
			\draw[red, line width=2.6] (-3.25,-4) -- (-1.25,-4); 
			\draw[teal, line width =2.6] (-3.25,-3.75) -- (-2.3,-3.75); 
			\draw[teal, line width =2.6] (-2.2,-3.75) -- (-1.25,-3.75); 
			
			\draw [draw=gray, fill=gray, opacity=0.2] (-3.25,-4.1) rectangle (-1.25,-3.65);
			
			\draw[red, line width=2.6] (-0.75,-4) -- (1.25,-4); 
			\draw[teal, line width =2.6] (-0.75,-3.75) -- (0.2,-3.75); 
			\draw[teal, line width =2.6] (0.3,-3.75) -- (1.25,-3.75); 
			\draw[teal, line width =2.6] (-0.2,-3.5) -- (0.75,-3.5); 
			
			\draw [draw=gray, fill=gray, opacity=0.2] (-0.75,-4.1) rectangle (1.25,-3.4); 			
			
		\end{tikzpicture}
		
		\caption{A fragment of a $T$-representation of a graph $G$, where the branch
			to the right is formed from a leaf edge of $T$.
			Notice that this branch carries many maximal
			cliques (actually $7$, but we can easily build many more there)
			which all except one can be leaf cliques, and
			so Procedure~\ref{proc:extract} finds their minimal separator consisting
			of the blue vertex and outputs the whole interval subgraph as one fragment.}
		\label{fig:TGraphx}
	\end{figure}

	\repeatlemma{reallycanonical}
	
	\begin{proof}
		One may easily verify that every step of Procedure~\ref{proc:extract}
		takes into account only isomorphism-invariant properties of the graph $G$,
		does not consider the input representation of $G$
		in any way and makes no arbitrary decisions.
		Consequently, every step performed by the procedure for the input $G$
		has an ``isomorphic'' step preformed for the input $G'$.
		This extends to possible recursive calls as well, and the conclusion follows.
		\qed \end{proof}

	\section{Additions to Section~\ref{sec:algebraic}}
	
	For further details regarding automorphism groups, see e.g., \cite{furst}.
	Here we briefly illustrate Babai's ``tower-of-groups'' procedure on the
	concrete example of Babai~\cite{babai-bdcm} (to which our use in
	Section~\ref{sec:mainalg} is conceptually very similar):
	A \emph{$d$-bounded color multiplicity graph} is a graph $G$ whose 
	vertex set is arbitrarily partitioned into $m$ color classes 
	$V(G)=V_1\cup\ldots\cup V_m$ such that $V_i \cap V_j = \emptyset$ for all $1 \leq i < j \leq m$.
	The number $m$ of colors is arbitrary, but for all $1\leq i\leq m$, the
	cardinality $|V_i|$, called the multiplicity of $V_i$, is at most~$d$. 
	To compute the automorphism group of such $G$, we start with $\Gamma_0$
	which freely permutes each color class of $G$ (formally, it is the product
	of the symmetric groups on each~$V_i$).
	Then, stepwise, we add the restrictions to preserve the edges between
	$V_i$ and $V_j$ for $(i,j)=(1,2),\>(1,3),\ldots,(1,m),\>(2,3),\ldots,(m-1,m)$.
	The last group $\Gamma_k$ for $k={m\choose2}$ is the automorphism group of~$G$
	and the total runtime is in \emph{FPT} with the parameter~$d$.
	
	We return in a closer detail to the crucial correspondence between
	automorphisms of $T$-graphs and automorphisms of our special decompositions
	as defined in Section~\ref{sec:algebraic}.
	Recall that a permutation $\varrho$ of $\ca X\cup\ca A$ is an {\em automorphism of the
		decomposition of~$H$} if the following hold true;
	\Aconditions
	In regard of \ref{it:autombetw} we remark that the words `corresponding
	attachment set' refer to the fact that attachment sets of $X$ are uniquely
	determined in the graph isomorphism to~$\varrho(X)$.
	
	\repeatproposition{autconsistent}
	
	\begin{proof}
		In the `only if' direction, we consider a permutation $\varrho$ of
		$\ca X\cup\ca A$ satisfying conditions \ref{it:automlev} and
		\ref{it:autombetw}, i.e., an automorphism $\varrho$ of the considered decomposition of~$H$.
		We take the mapping $\pi$ on $V(H)$ which is composed of all isomorphism
		bijections from $X\in\ca X$ to $\varrho(X)$ claimed by \ref{it:automlev}.
		Then $\pi$ indeed is a permutation of $V(H)$ since $\varrho$ is a
		permutation on~$\ca X$, and $\pi$ respects all edges of $E(H)$ which belong to some $X\in\ca X$.
		All remaining edges of $H$ are between some fragment $X\in\ca X_i$ and one
		of its attachments (``higher up'' in the decomposition) which coincides with
		some terminal set $A\in\ca A_j$ where $j>i$, by the way we decomposed~$H$.
		Condition \ref{it:autombetw} ensures that those edges of $H$ are preserved
		as well by $\pi$, and hence $\pi$ is an automorphism of~$H$.
		See Figure~\ref{fig:terminalAut}.
		
		In the `if' direction, by recursive application of Lemma~\ref{lem:reallycanonical},
		we get that
		\begin{itemize}
			\item any automorphism of $G_1$ (or of $G_2$, up to symmetry) preserves the
			fragments and the levels of the decomposition of $G_1$ and, consequently, it
			preserves also the terminal sets by their incidence with attachments of the
			fragments;
			\item if $G_1\simeq G_2$, we have an isomorphism $\iota:G_1\to G_2$ preserving 
			the fragments and the levels between $G_1$ and $G_2$, and then $\iota$ can
			be composed with any automorphism of $G_2$ from the previous point.
		\end{itemize}
		Consequently, for every graph automorphism $\sigma$ of $H$ we get an induced permutation on
		$\ca X\cup\ca A$, which indeed is an automorphism of the decomposition
		according to the conditions \ref{it:automlev} and \ref{it:autombetw} --
		simply because $\sigma$ was a graph automorphism.
		\qed\end{proof}

	\section{Automorphisms of set families (or of hypergraphs)}
	\label{sec:automvenn}
	
	In order to efficiently compute with terminal sets introduced by
	Procedure~\ref{proc:decomprec}, we give the following technical result from
	\cite{aaolu2019isomorphism}.
	
	Let $\ca U$ and $\ca U'$ be set families over finite ground sets $Z$ and
	$Z'$ (with no additional structure), respectively.
	A~bijection $\pi$ from $\ca U$ to $\ca U'$ is called an {\em isomorphism}
	if and only if there exists a related bijection $\zeta$ from $Z$ to $Z'$ such that
	$\pi$ and $\zeta$ together preserve the incidence relation~$\in$, i.e.,
	for all $U\in\ca U$ and $z\in Z$ we have $z\in U$ $\iff$ $\zeta(z)\in \pi(U)$.
	This is essentially the same concept as that of an isomorphism between hypergraphs
	$(Z,\ca U)$ and $(Z,\ca U')$, but notice that we primarily focus on the
	mapping between the sets of $\ca U$ and $\ca U'$, and not on the mapping
	between the elements of $Z$ and $Z'$.
	An isomorphism $\pi$ from $\ca U$ to $\ca U$ is an {\em automorphism of~$\ca U$}.
	
	For any set family $\mathcal{U}$, we call
	a {\em cardinality Venn diagram} of $\mathcal{U}$ the vector
	$\big(\ell_{\mathcal{U},\mathcal{U}_1}: \emptyset\not=\mathcal{U}_1\subseteq\mathcal{U}\big)$
	such that $\ell_{\mathcal{U},\mathcal{U}_1}:=|L_{\mathcal{U},\mathcal{U}_1}|$ where
	\smallskip
	$L_{\mathcal{U},\mathcal{U}_1}=\bigcap_{A\in\mathcal{U}_1}\!A
	\setminus \bigcup_{B\in{\mathcal{U}\setminus\mathcal{U}_1}}\!B$.
	That is, we record the cardinality of every internal cell of the Venn diagram
	of~$\mathcal{U}$.
	Let $\pi(\mathcal{U}_1)=\{\pi(A):A\in\mathcal{U}_1\}$
	for $\mathcal{U}_1\subseteq\mathcal{U}$.
	
	The following is a straightforward but crucial observation:
	\begin{repproposition}{checkvenn}
		\label{pro:testingvenn}
		For a set family $\ca U$ over $Z$, a permutation $\pi$ of $\ca U$ is an
		automorphism of~$\ca U$ if and only if the cardinality Venn diagrams of $\mathcal{U}$ and of 
		$\pi(\mathcal{U})$ are the same, meaning that
		$\ell_{\mathcal{U},\mathcal{U}_1}=  \ell_{\mathcal{U},\pi(\mathcal{U}_1)}$
		for all~$\emptyset\not=\mathcal{U}_1\subseteq\mathcal{U}$.
		\\The latter condition can be tested in time $\ca O(m^2)$ where~$m=|\ca U|+|Z|$.
		
	\end{repproposition}
	
	Notice that the problem of computing the automorphism group of such set
	family $\ca U$ is GI-complete in general -- we can take $\ca U$ as the set
	of edges of a graph as 2-element subsets.
	Nevertheless, we can compute the group efficiently in the special case
	of our terminal sets which have bounded-size antichains:
	
	\begin{repproposition}{setautom}
		\label{lem:setautom}
		{\bf($\!\!${\cite[Algorithm~2 and Lemma~5.9]{SdT-graphs2021efficient}})}\label{lem:vennautom}
		Let $\ca U$ be a set family over a finite ground set $Z$ such that the
		maximum size of an antichain in $\ca U$ is~$d$ (i.e., there are no more than
		$d$ sets in $\ca U$ pairwise incomparable by the inclusion).
		Then the automorphism group of $\ca U$ can be computed in \emph{FPT}-time
		with respect to~$d$.
	\end{repproposition}
	
	To give a brief sketch of a proof here, we observe the following \cite[Lemma~5.7]{aaolu2019isomorphism}:
	If a permutation $\pi$ on $\ca U$ fails to be an automorphism of $\ca U$, then
	there is a subfamily $\ca U_2\subseteq\ca U$ witnessing this failure
	such that $\ca U_2$ is an antichain in the inclusion.
	Together with the bound $|\ca U_2|\leq d$ on antichains, this is enough to make 
	Babai's tower-of-groups machinery work in \emph{FPT}-time.

	\section{Automorphisms of Interval Graphs with Marked Sets}
	
	The task here is to argue that the algorithm given in
	Theorem~\ref{thm:mainautom} computes efficiently in Step~\ref{it:step1b}, items (a) and
	(c), before finishing the full proof in the next section.
	Recall that the task is to compute the automorphism group of an interval
	graph~$G$ (which is easy in the basic setting~\cite{AutMPQtrees}), but under an
	additional constraint that a given set family $\ca A\subseteq2^{V(G)}$ (recall the terminal sets) is
	preserved by the automorphisms -- that each set from $\ca A$ is mapped into a set from $\ca A$.
	
	The latter problem is generally GI-hard since $G$ may be chosen as a clique
	and $\ca A$ as the edge set of an arbitrary graph $H$, but the crucial
	restriction in our case is that the maximum size of an antichain of sets in
	$\ca A$ is bounded (cf.~Proposition~\ref{lem:setautom}).
	Then we obtain:
	
	\begin{lemma}\label{lem:intterminals}
		Let $G$ be an interval graph and $m>0$ an integer.
		Let $\ca A^1,\ldots,\ca A^m$ be families (in general multisets) of subsets of $V(G)$ (terminal or {\em marked} sets of~$G$) such that,
		for $\ca A:=\ca A^1\cup\ldots\cup\ca A^m$,
		\begin{itemize}
			\item every set $A\in\ca A$ induces a clique of~$G$, and
			\item the maximum size of an antichain in $\ca A$ equals $t$
			(i.e., there are no more than
			$t$ sets in $\ca A$ pairwise incomparable by the inclusion).
		\end{itemize}
		Denote by $\Gamma_1$ the group consisting of those automorphisms $\sigma$ of $G$
		such that, for each $i\in\{1,\ldots,m\}$,  $\sigma$ preserves the set family $\ca A^i$.
		Then one can in \emph{FPT}-time with the parameter $t$ (but independently of~$m$) 
		compute the group $\Gamma$ of permutations of $\ca A$ which is the action of $\Gamma_1$ on $\ca A$.
		In more detail, a permutation $\tau$ of $\ca A$ belongs to $\Gamma$, if and only if there exists an automorphism $\varrho$ of $G$
		such that, for every $i\in\{1,\ldots,m\}$ and all $A\in\ca A^i$, we have~$\tau(A)\in\ca A^i$ and~$\tau(A)=\varrho(A)$.
	\end{lemma}
	
	\begin{proof}[sketch]
		All interval representations of an interval graph, or equivalently all its clique
		paths (recall the clique trees of chordal graphs from Section~\ref{sec:decom}), 
		can be represented by one suitable data structure called the
		{\em PQ-tree}~\cite{recogIntervalLinear}.
		A PQ-tree $T$ is a rooted ordered tree whose internal nodes are labelled as either
		P-nodes or Q-nodes, where the children of a P-node can be arbitrarily
		reordered, while the order of the children of a Q-node can only be reversed.
		The leaves of $T$ hold maximal cliques of our interval graph $G$.
		The following fact is crucial~\cite{recogIntervalLinear,AutMPQtrees}:
		\begin{itemize}
			\item[(*)] For every interval graph $G$ one can in linear time construct
			a PQ-tree $T$ (with leaves in a bijection with the maximal cliques of $G$),
			such that the permissible reorderings of $T$ are in a one-to-one
			correspondence with all interval representations of~$G$.
		\end{itemize}
		In particular, this means that every automorphism of $G$ can be represented
		as a permissible reordering of $T$ (though, not the other way round).
		
		Every node $p$ of a PQ-tree $T$ of $G$ can be associated with a subgraph of
		$G$ formed by the union of all cliques of the descendant leaves of~$p$
		-- this subgraph {\em belongs} to~$p$.
		Then for a node~$p$ we define the {\em inner vertices of~$p$} as those
		vertices of $G$ which belong to $p$ and to at least two child nodes of $p$,
		but they do not belong to any sibling node of $p$.
		(In the case of a P-node, the inner vertices of $p$ belong to all child nodes
		of $p$, but this is generally not true for Q-nodes.)
		
		The first step is to reduce the tree $T$ into a small subtree which is
		``essential'' for the sets of~$\ca A$.
		Precisely, call a node $p$ of $T$ {\em clean} if the inner vertices of $p$
		are disjoint from $\bigcup\ca A$.
		The subtree rooted at $p$ is then {\em clean} if $p$ and all descendants of
		$p$ in $T$ are clean.
		Observe that the number of non-clean subtrees at the same depth of $T$ is
		always bounded from above by~$t$.
		Indeed, since every set $A\in\ca A$ induces a clique in $G$, at most one of
		the considered non-clean subtrees can be caused by the same set $A$, and the
		witnessing sets of the non-clean subtrees form an antichain.
		Now, for every node $q$ of $T$, we use \cite{AutMPQtrees} to determine the
		isomorphism class of the subgraph belonging to the PQ-tree formed by $q$ and its clean subtrees.
		We store this information as an annotation of $q$ and discard the clean subtrees.
		Let the resulting reduced tree be denoted by~$T'\subseteq T$.
		
		We show that the automorphisms of the reduced tree $T'$ (with the
		aforementioned annotation) can be handled using the tools from
		Section~\ref{sec:automvenn}.
		By~(*), we can in a canonical (i.e., automorphism-invariant) way decompose 
		the vertex set of $G$ into layers; where layer $i$ is formed 
		by the inner vertices of the nodes of $T$ which are at depth~$i$.
		In the subsequent argument, we show that structure of the tree $T'$ can
		now be ``replaced'' by suitably chosen sets added to the terminal set family~$\ca A$.
		For each node $q$ of $T'$ we, essentially, add the set $B_q$ formed by the
		vertices of the subgraph of $G$ belonging to $q$ in~$T'$.
		This does not increase the maximum antichain size, since $B_q$ in an
		incomparable subfamily can be replaced by any set of $\ca A$ contained in~$B_q$.
		We also keep the annotation of $q$ as an annotation of the set $B_q$.
		Moreover, for a Q-node $q$, we annotate the order of the children of $q$ (in
		a reversible way).
		Finally, we can use Proposition~\ref{lem:setautom} to compute the
		(annotation-preserving) automorphism group of the resulting set family, which
		coincides with the desired permutation group $\Gamma$ when restricted onto~$\ca A$.
		
		Since the previous claims are interesting on their own, we leave
		the full detailed description for a separate paper~\cite{A2021autmarked}.
		\qed\end{proof}

	\begin{corollary}\label{cor:intterminals}
		Let, for $j=1,2$, $G_j$ be a connected interval graph, $m>0$ an integer, 
		$t$~a~parameter, and $\ca A_j^1,\ldots,\ca A_j^m$ families of subsets of $V(G_j)$,
		all as in Lemma~\ref{lem:intterminals} for each~$j\in\{1,2\}$.
		Then one can in \emph{FPT}-time with the parameter $t$ decide whether there
		exists an isomorphism from $G_1$ to $G_2$ 
		bijectively mapping $\ca A_1^i$ to $\ca A_2^i$ for all~$1\leq i\leq m$.
	\end{corollary}
	\begin{proof}
		We simply consider the interval graph $G$ formed as the disjoint union of
		$G_1$ and $G_2$, compute the respective permutation group on 
		$\ca A_1\cup\ca A_2$ exactly as in Lemma~\ref{lem:intterminals}
		and check whether some permutation exchanges $\ca A_1$ with $\ca A_2$.
		\qed\end{proof}
	
	We, moreover, remark that the condition in Step~\ref{it:step1b} of
	Theorem~\ref{thm:mainautom} -- namely that we require the
	isomorphisms\,/\,automorphisms to preserve the tail of $X^+$,
	can easily be respected in Lemma~\ref{lem:intterminals} by introducing 
	a separate family of a single set with the tail(s).

	\section{Additions to Section~\ref{sec:mainalg}}
	
	We finish with the skipped detailed arguments for Section~\ref{sec:mainalg}.
	\repeatcorollary{cormain}
	
	\begin{proof}
		Let the promised leafage of $G_1$ and $G_2$ be at most $d$.
		We exhaustively try all trees $T_1,T_2,\ldots$ without degree-$2$ vertices
		and with at most $d$ leaves; their total number depends only on $d$
		(exponentially).
		For each such $T_i$ sequentially, we call the algorithm of Theorem~\ref{thm:main} 
		with $T=T_i$, and if we ever get an answer about $G_1\simeq G_2$, we output it
		and quit.
		If all answers are that $G_1$ or/and $G_2$ are not $T_i$-graphs, then the
		promise of leafage $\leq d$ is violated.
		\qed \end{proof}
	
	\repeattheorem{mainautom}
	
	\begin{proof}
		We refer to the algorithm outline in the main paper.
		Correctness of Steps~\ref{it:step1b} and~2 with respect to the condition \ref{it:automlev} is self-evident.
		Regarding efficiency of computation in Step~\ref{it:step1b} we refer to
		Lemma~\ref{lem:intterminals} and Corollary~\ref{cor:intterminals}.
		The rest of these steps follows by standard computation with groups 
		(which are represented by their sets of generators, as usually).
		
		Correctness of Step~\ref{it:Btower} is again self-evident -- we stepwise ensure that the
		resulting subgroup $\Gamma_m$ satisfies by all its members the condition
		\ref{it:autombetw}, which together with aforementioned \ref{it:automlev}
		imply that $\Gamma_m$ indeed is the automorphism group of the given
		decomposition of~$H$.
		
		\smallskip
		We now add more details to justification of proper and efficient use of
		Theorem~\ref{thm:furstgen} is Step~\ref{it:Btower}.
		In particular, this meas to show that the ratio $|\Gamma_{k-1}|/|\Gamma_k|$
		is bounded in the parameters $s,t$ in order to claim runtime in \emph{FPT}.
		For the latter, consider any two permutations $\varrho,\sigma\in\Gamma_{k-1}$
		which mutually agree on mapping of all fragments and terminal sets considered
		by Step~\ref{it:Btower} in the current iteration~$k$.
		Hence the composed permutation $\sigma^{-1}\circ\varrho$ is identical on the
		elements currently considered by Step~\ref{it:Btower}, and the condition \ref{it:autombetw}
		in the current iteration $k$ is automatically true for
		$\sigma^{-1}\circ\varrho$, meaning that $\sigma^{-1}\circ\varrho\in\Gamma_k$.
		Consequently, such $\varrho$ and $\sigma$ belong to the same coset of the
		subgroup $\Gamma_k$ in $\Gamma_{k-1}$ by the definition.
		It is well known that the number of these cosets equals $|\Gamma_{k-1}|/|\Gamma_k|$
		which we would like to estimate.
		
		Now, how many permutations in $\Gamma_{k-1}$ are there that pairwise
		disagree on mapping of all components and terminal sets considered by Step~\ref{it:Btower}?
		In iteration $k$ of Step~\ref{it:Btower} we have at most $s$ components to be mapped on
		level $i'$, at most $t$ terminal sets of cardinality $r$ in every fragment
		(since these sets form an antichain).
		This gives a possibility of at most $s!\cdot(st)!$ distinct mappings,
		and hence $|\Gamma_{k-1}|/|\Gamma_k|\leq s!\cdot(st)!$ as needed by Theorem~\ref{thm:furstgen}.
		\qed\end{proof}
	
	Finally, regarding the overall \emph{FPT}-runtime of the whole algorithm
	composed of Procedures \ref{proc:extract}, \ref{proc:decomprec},
	and Lemma~\ref{lem:intterminals} and Theorem~\ref{thm:furstgen}.
	The first two procedures run in polynomial time regardless of our
	parameters, and the latter two algorithms take each \emph{FPT}-time with
	respect to parameters which are bounded by functions of the ``master''
	parameter equal to the number of leaves of our fixed tree~$T$.

\end{document}